\documentclass{article}
\usepackage[utf8]{inputenc}
\usepackage{amsmath}
\usepackage{amsthm}
\usepackage{amssymb}
\usepackage{mathrsfs}
\usepackage{bbm}
\usepackage{hyperref}
\usepackage[authoryear]{natbib}
\usepackage[margin=1in]{geometry}
\usepackage{setspace}
\usepackage{xcolor}
\usepackage{ulem}
\usepackage{authblk}
\usepackage{fancyhdr} 
\usepackage{graphicx}
\usepackage{tikz}
\usepackage{comment}
\usepackage{csquotes}
\usetikzlibrary{patterns}
\usepackage{caption}
\usepackage{subcaption}
\newcommand{\minus}{\scalebox{0.75}[1.0]{$-$}}

\newtheorem{proposition}{Proposition}
 
\title{General Equilibrium Effects of Carbon Offsets}
\author[1]{Isla Globus-Harris}
\author[2]{Daniel H. Karney}
\affil[1]{Colgate University\footnote{iglobusharris@colgate.edu, 13 Oak Drive, Hamilton, NY 13346, USA.}}
\affil[2]{Ohio University\footnote{karney@ohio.edu, 1 President St., Athens, OH 45701, USA.}}

\date{\today}

\begin{document}

\pagestyle{fancy}
\fancyhf{}
\fancyhead[R]{\thepage}
\fancyhead[C]{General Equilibrium Effects of Carbon Offsets}

\maketitle
\thispagestyle{empty}


\begin{abstract}
\noindent We construct an analytical general equilibrium model of an economy with carbon offsets, and show that increasing the carbon offset price has an ambiguous effect on aggregate emissions and welfare. Using two carbon accounting metrics, we demonstrate that offsets are over-credited under many parameterizations; however, offset under-crediting can also occur. Due to general equilibrium effects, neither carbon accounting metric is a sufficient statistic for welfare. Furthermore, we define four margins whereby offsets can respond to payments, including a margin not previously identified. Our results suggest that market spillover effects warrant consideration when evaluating carbon offset policies.
\end{abstract}



\vspace{5.5cm}

\noindent We are grateful for comments and suggestions from Teevrat Garg, Rick Klotz, and Jeffrey Wagner, from participants at the 2025 AERE Summer Conference, 2025 AERE@SEA Sessions, 2026 AERE@MEA Sessions, and the 2026 World Congress of Environmental and Resource Economists (WCERE), as well as seminars at Colgate University, Miami University (OH), and Smith College. We thank Elly Beauchesne for excellent research assistance. All mistakes are our own. No AI was employed in writing of this manuscript or any of the analysis therein; however, Gemini was used for troubleshooting and editing Matlab code.

\newpage
\setcounter{page}{1}

\section{Introduction}

Demand for carbon offsets from policymakers, corporations, and nonprofits is growing around the world \citep{harrison2024carbon}. As the size of the offset market grows, there will be increased impacts of offset policies on other markets, emissions, and the broader economy. However, the general equilibrium effects of carbon offsets are not well studied. Furthermore, the (non-)additionality of carbon offsets is a well-known problem \citep{bushnell2010economics, zhang2011co, calel2025carbon, aspelund2026}. In this study, we examine the economy-wide impacts of carbon offsets using a general equilibrium framework, explore how general equilibrium effects limit the accuracy of offset carbon accounting metrics with respect to actual emission changes, and show that neither direct nor aggregate emission changes are sufficient statistics for welfare change. We also investigate the consequences of offset (non-)additionality in general equilibrium.

We address these issues by constructing a tractable general equilibrium model of an economy with carbon offsets. The model has two types of inputs---clean and dirty---and the clean input can be subsidized by carbon offset payments. For simplicity, we refer to these inputs as renewable (clean) energy and fossil fuel (dirty) energy, but our model can be applied to non-energy settings such as agricultural carbon offsets. Consumers gain utility from a composite final consumption good, but lose utility from the emissions externality. Emissions are produced by both the fossil fuel and final good sectors. We solve for closed-form solutions and find that raising the carbon offset price increases clean energy production while decreasing final good production. However, dirty energy production can increase and aggregate emissions can increase too. Although a reduction in aggregate emissions is a necessary condition for an increase in welfare, it is not sufficient. Rather, welfare effects depend on the relative size of the consumption good decrease (i.e., the policy's cost) versus the value of the emissions change.

We also decompose carbon offsets into four types of additional and non-additional offsets. Additionality is the idea that the emissions reductions from an offset-generating project would not have occurred in the absence of the offset payment. We describe the margins upon which additional and non-additional offsets respond to price changes, one of which has not previously been identified in the literature.  Then, we assess the model's results using two separate carbon accounting metrics: conventional (i.e., direct) carbon accounting and aggregate carbon accounting. Conventional carbon accounting credits offsets with the direct emissions reductions they cause and aligns with most practitioners' definitions. Aggregate carbon accounting is the change in total, economy-wide emissions that includes both direct and indirect changes. We analytically and numerically demonstrate the difference between the two metrics and how that difference changes as key parameters shift. Importantly, we highlight a type of rebound effect where conventional carbon accounting often over-credits offsets when compared to the aggregate emissions decreases (due to market and economy-wide responses). However, we also document cases where conventional carbon accounting under-credits carbon offsets. Yet, even if we assume that all offsets are real and additional---as measured by traditional screening tools---our results suggest a further problem with carbon accounting: real and additional offsets can increase aggregate emissions. 

In the economics literature, the term ``additionality'' is used in ways that can align with our two carbon accounting metrics and definitions; thus, we introduce more precise terminology for discussing offset additionality and its measurement. \cite{calel2025carbon} examine additionality of wind projects in India at the project-level, determining that at least 52 percent of the projects are non-additional. In our context, these projects are called ``non-additional by traditional, project-based assessments'' and their emissions reductions are not credited using our conventional carbon accounting metric. Other economists have taken a system-wide view of additionality in line with our aggregate carbon accounting metric. For example, \cite{zhang2011co} examine the additionality of power sector Clean Development Mechanism (CDM) projects in China using economy-wide emissions (similar to our aggregate carbon accounting metric), since project-specific emissions data were not available. \cite{bushnell2010economics} is one of the few studies that makes the distinction between project-level and aggregate impacts, albeit in a less formal manner than our treatment. 

There is a small but growing literature that empirically assesses carbon offset additionality \citep{schneider2011perverse, calel2025carbon, aspelund2026}, and a small theoretical literature on offset additionality, with the latter focusing on screening and other solutions to additionality concerns \citep{ montero2000optimal,mason2013additionality, vanbenthem2013scale,globus2020impossible}. Instead, we focus on precisely defining two distinct carbon accounting metrics, as well as describing several different margins upon which carbon offsets can respond to price changes. These additionality and measurement concerns are connected to a larger literature on carbon policy leakage and that literature focuses more broadly on emissions leakage under carbon taxes, cap-and-trade systems, or other regulations (e.g., \citet{mcguinness2008effects, fowlie2016market, guo2023cost, wang2024trade}. A more limited literature examines leakage in the context of carbon offsets; for example, \citet{filewod2023avoiding} discuss various types of emissions leakage including what they call ``market'' leakage that operates via price effects. The bulk of this literature consists of empirical work quantifying the extent of offset leakage \citep{alix2012forest, heilmayr2020deforestation}. We extend the literature on carbon offset leakage by examining offsets and their net impact on emissions in a general equilibrium model while accounting for all possible channels of market-based leakage.

We also contribute to the literature that uses general equilibrium models to examine environmental policies. Previous studies use these models to examine pollution taxes \citep{fullerton2007general,ren2011optimal, garnache2022environmental}, renewable portfolio standards \citep{bento2018emissions, fullerton2025determines}, emissions trading schemes \citep{konishi2015emissions}, leakage \citep{bohringer2025measures, karney2025model}, and environmental mandates \citep{fullerton2010general, hafstead2018unemployment, fullerton2020costs}. We expand this literature by examining the general equilibrium impacts of carbon offset policies on carbon accounting and emissions.

The carbon offset payments in our model can be interpreted as a type of abatement subsidy or clean input subsidy. In the simplest partial equilibrium setting, such a subsidy can be equivalent to a Pigouvian tax \citep{baumol1988theory}; however, in a more general setting, there are several important differences between an abatement subsidy and a Pigouvian tax on carbon emissions \citep{sallee2025trouble}. First, in our model, carbon offset payments only partially subsidize abatement due to regulatory constraints: producers are paid for producing more renewable energy (thereby avoiding emissions), but are not compensated for reducing emissions in other sectors, making the carbon offset payments an incomplete abatement subsidy. Second, to isolate the general equilibrium effects of carbon offsets in the economy, we specify that the government is the sole purchaser of offsets. As a result, the government must generate revenue for the offset payments.\footnote{This revenue generation effect on welfare has similarities to the double dividend literature (see \cite{bento2024environmental} for an overview of this literature and its history): the double dividend refers to the possibility that a revenue-generating Pigouvian tax can correct externalities while also decreasing the burden of a distortionary tax system. We see the reverse, carbon offset payments are a type of subsidy requiring revenue generation. Relatedly, the zero-profit conditions in our model mean that voluntary offset payments cannot occur.} 
Third, while the revenue-generating taxes in our model are equivalent to lump-sum taxes and thus a pure transfer, implementing the offset payments still incur a real welfare cost via lost consumption. In particular, the offset payments change relative prices in the economy and so introduces a distortion. Therefore, carbon offsets partially correct an externality while introducing a distortionary cost. These costs can be seen in our welfare analysis: carbon offsets may decrease emissions, but they also decrease final good consumption.\footnote{If raising revenue did distort behavior, then this would be an additional cost of the offset policy \citep{bovenberg1994environmental}. However, the lump-sum tax assumption isolates only the effect of the offset policy rather than simultaneously considering distortions from the tax structure that raises revenue. Thus, our analysis provides a clean baseline for understanding the welfare effects of offset policies.}

As a numerical exercise to illustrate the range of outcomes, we parameterize our model using plausible parameter values from the U.S. economy. For a large set of cases, we find that carbon offsets are over-credited because they cause aggregate emissions to decrease by less than the amount credited using conventional carbon accounting. This is true even when offset projects would not have occurred in the absence of the offset payments---that is, the offsets are real and additional by traditional project-based measures---and arises due to downstream effects on other quantities in the economy. Under certain conditions, carbon offsets can actually \textit{increase} carbon emissions due to these economy-wide effects. This outcome is similar to that found in the rebound effect literature, where an extreme rebound effect (referred to as ``backfire'') means that an improvement in energy efficiency leads to an increase in energy use.\footnote{The rebound effect is when an improvement in energy efficiency causes an increase in use of a durable good, as its marginal operating cost has decreased \citep{gillingham2016rebound}.} 
However, in our setting, the mechanism is different. The traditional rebound effect is driven by a change in marginal cost and is therefore most strongly seen in partial equilibrium while tending to be muted in general equilibrium \citep{fullerton2020costs}. Our backfire effect is present only in general equilibrium due to relative price changes.

Conversely, we find that carbon offsets can also be under-credited; that is, the real emission reductions can be greater than those calculated under conventional carbon accounting. 
As an example, suppose that a significant proportion of wind turbines receiving offset payments would have been built without an offset policy and are considered non-additional for carbon accounting purposes \citep{calel2025carbon}. If there is an increase in the offset price after the turbines were built, then the profit maximizing operators respond to this marginal incentive by increasing output, perhaps by reducing downtime for maintenance.\footnote{\citet{davis2012deregulation} find that market-based incentives from deregulation reduces outages at nuclear power plants, while \citet{karney2019electricity} establishes that output at these plants increases due to market-wide environmental policy.} 
However, if the increase in output (and corresponding emissions decrease elsewhere) from these traditionally non-additional sources is not counted as an effect of an offset price increase, then the conventional carbon accounting metric under-credit offsets relative to the actual emission reductions. This is the additionality margin not previously identified in the literature.

Our findings have important implications for carbon accounting and offset policy, since these policies can have substantial general equilibrium effects due to the (potential) size of the carbon market. In particular, policymakers should be aware that even real and additional carbon offsets have spillover effects that may erode their emissions benefits, and in the extreme case, can even increase aggregate emissions. Practitioners may wish to re-examine how they credit carbon offsets to account for these market spillover effects. Furthermore, we detail how conventional carbon accounting may be under-crediting offsets and, generally, is a poor proxy for welfare. In this way, we add to the literature on the limitations of carbon offsets and carbon accounting metrics as a tool to expand the reach and reduce the costs of incomplete carbon regulation.

The remainder of the paper is organized as follows. Section~\ref{Sec:Model} sets out the general equilibrium model. Section~\ref{Sec:Additionality} discusses and formalizes our definitions of conventional and aggregate carbon accounting, and delineates our taxonomy of offsets and their margins of adjustment. Section~\ref{Sec:Welfare} provides definitions of welfare and the marginal value of public funds (MVPF), and discusses optimal policy in this setting. Section~\ref{Sec:Analytical} presents the analytical solution to the model, and Section~\ref{Sec:Numeric} provides illustrative numerical examples to highlight the relative importance of different parameters, as well as a case study of carbon capture and storage subsidies. Section~\ref{sec:conclusion} concludes and discusses policy implications.


\section{Model} \label{Sec:Model}

We model an economy with two types of energy---a single final consumption good and an emissions externality ---following \citet{fullerton2025determines}.\footnote{\citet{fullerton2025determines} examine renewable portfolio standards, while we examine carbon offsets. Our model can be interpreted as applying to renewable energy credits, which are a component of many renewable portfolio standard policies. Our model also includes a fixed factor component in the renewable energy sector, which adds additional complexity that reflects real-world resource constraints and is necessary for tractability.} We extend their model to include carbon offset payments to renewable energy producers. Total productive resources in the economy are represented by an aggregate of capital, labor, and materials. This total resource is fixed and denoted by $\bar K$. We call this aggregate resource simply ``capital'' for convenience, but it can be more generally defined as a composite of the primary inputs to production. Capital is employed to produce three goods: an aggregate, final consumption good ($X$); renewable energy ($R$); and fossil fuel energy ($F$). Although we call these inputs ``energy'', the model works for any clean and dirty inputs; for example, $R$ could be low-carbon emission agriculture (e.g., no-till) while $F$ is conventional agriculture, both used to produce the consumption good (food). There are many other possible examples that fit the model, including: landfills with and without methane capture technology, traditional cookstoves versus clean cookstove technology, forestry services, and livestock farms with and without biogas technology.

The resource constraint in the economy is given by $\bar K = K_X+K_R+K_F$ where $K_i$ denotes the amount of capital employed in sector $i \in \{X,R,F\}$. The additive form of the resource constraint implies capital is homogeneous and thus paid a uniform price ($P_K$). We log-linearize the economy as shown in Appendix \ref{App:LogLinear}. The log-linear version of the resource constraint becomes:
\begin{align} \label{Eq:Resource}
    0 &=\alpha_X \hat K_X + \alpha_R \hat K_R + \alpha_F \hat K_F,    
\end{align}
where the ``hat'' notation indicates a small, marginal change; specifically, $\hat{K}_i \equiv \frac{d K_i}{K_i}$. The parameters $\alpha_i \equiv \frac{K_i}{\bar K}$ in equation \ref{Eq:Resource} express the initial share of capital employed in each production sector, and thus $1 = \alpha_X + \alpha_R + \alpha_F$. Therefore, an increase in capital use in the renewable energy sector $(R)$ requires a decrease in capital use in the fossil fuel energy sector $(F)$ and/or final goods sector $(X)$. As discussed in Section \ref{Sec:Welfare}, the potential to decrease capital use in the final goods sector is important because consumers receive their consumption utility directly from good $X$.

The final consumption good ($X$) is produced via a constant returns to scale (CRTS) production function using capital ($K_X$) and total energy ($E$) as inputs given by $X = X(K_X, E)$, and sold at price $P_X$. Total energy is the combination of fossil ($F$) and renewable ($R$) energy produced by an energy sector $E$ then sold at price $P_E$ to sector $X$.\footnote{Alternatively, $E$ could be produced within sector $X$ via a nested sub-production function.} Totally differentiating sector $X$'s production function gives the condition:
\begin{align}
    \hat X &= \theta_{XK}\hat K_X+\theta_{XE}\hat{E}, \label{Eq:ProductionX}
\end{align}
where $\hat X$ is the change in total final goods production $(\hat X \equiv \frac{dX}{X})$, and the other ``hat'' variables are similarly defined. The $\theta_{Xj}$ parameters indicate the initial input costs shares such that $\theta_{XK} \equiv \frac{(1+t)P_KK_X}{P_XX}$ where $t$ is an \textit{ad valorem} tax on capital (equivalent to a lump-sum tax since $\bar K$ is fixed and all capital is sold into the input market). This means $(1+t)P_K$ is the per unit cost of capital for firms; the capital tax revenue will be used to pay for carbon offsets. Similarly, $\theta_{XE} \equiv \frac{P_EE}{P_XX}$ is the input cost share for energy and the cost shares must sum to one due to the CRTS assumption meaning $1 = \theta_{XK} + \theta_{XE}$. The CRTS assumption also implies a zero-profit condition for sector $X$ given by $P_X X = (1+t)P_K K_X + P_E E$; totally differentiating yields:
\begin{align}
    \hat P_X +\hat{X} &= \theta_{XK} (\hat t + \hat P_K + \hat K_X) + \theta_{XE} ( \hat P_E+ \hat E), \label{Eq:ZeroProfitX}
\end{align}
where $\hat t \equiv \frac{dt}{1+t}$ is the tax change.\footnote{Equation \ref{Eq:ZeroProfitX} can be represented as $\hat P_X = \theta_{XK} (\hat t + \hat P_K) + \theta_{XE} \hat P_E$ after subtracting equation \ref{Eq:ProductionX}. In this form, equation \ref{Eq:ZeroProfitX} says that the change in the output price is solely a function of the changes in the inputs prices and tax rate. The log-linear zero-profit condition for each of the other production sectors can be rewritten similarly.} Substitution in production between the capital and energy inputs in the final goods sector is summarized by the elasticity of substitution $\sigma_X \geq 0$. Applying the elasticity's definition leads to:
\begin{align}
    \hat K_X-\hat E &= \sigma_X (\hat P_E - \hat t - \hat P_K), \label{Eq:ElastX}
\end{align}
such that increasing the energy price $(\hat P_E>0)$, all else equal, increases the ratio of capital to energy $(\hat K_X-\hat E>0)$ when $\sigma_X > 0$, and, conversely, the capital-to-energy input ratio falls if the price of capital or the capital tax increases. As $\sigma_X$ goes to infinity the inputs tend towards perfect substitution, while $\sigma_X=0$ implies perfect complements (i.e., Leontief production).

The energy sector's CRTS production function is given by $E = E(F,R)$, and differentiating finds:
\begin{align}
    \hat E &= \theta_{EF} \hat F + \theta_{ER} \hat R. \label{Eq:ProductionE}
\end{align}
The input share parameters in equation \ref{Eq:ProductionE} are defined as $\theta_{EF}\equiv \frac{P_FF}{P_EE}$ and $\theta_{ER}\equiv \frac{P_RR}{P_EE}$, where $P_F$ and $P_R$ are the prices for renewable and fossil fuel energy, respectively. The CRTS assumption means $1 = \theta_{EF} + \theta_{ER}$ and totally differentiating the zero-profits condition $(P_E E = P_F F + P_RR)$ recovers:
\begin{align}
    \hat P_E +\hat E &= \theta_{EF} ( \hat P_F+\hat F ) + \theta_{ER} ( \hat P_R+\hat R ). \label{Eq:ZeroProfitE}
\end{align}
Also, it is helpful to define $\theta_{XF} \equiv \frac{P_FF}{P_XX} = \theta_{XE} \theta_{EF}$ and $\theta_{XR} \equiv \frac{P_RR}{P_XX} = \theta_{XE} \theta_{ER}$ as the shares of dirty and clean energy input costs, respectively, in the production of the final consumption good; it follows that $1 = \theta_{XK} + \theta_{XF} + \theta_{XR}$. The elasticity of substitution in energy production, $\sigma_E \geq0$, summarizes substitution in production in the energy sector, and applying the definition yields:
\begin{align}
    \hat R-\hat F &= \sigma_E (\hat P_F - \hat P_R). \label{Eq:ElastE}
\end{align}
As $\sigma_E>0$ increases, then the different types of energy become more substitutable, while $\sigma_E=0$ implies that renewable and fossil fuel energy are prefect complements in the energy sector $(E)$. 

The production of renewable energy $(R)$ requires the primary factor, capital, as well as a sector-specific resource denoted $\bar{Q}$. (We delineate additional and non-additional $R$ later in this section.) The sector-specific resource enters the CRTS production function, $R = ( K_R,\bar{Q})$, as a fixed factor such that $d\bar{Q}=0$. The fixed factor accounts for the limited locations suitable for building renewable energy projects such as wind and solar farms (e.g., suitable windy locations and suitable sunny locations).\footnote{$Q$ could also represent scarce non-renewable resources that are critical to the clean energy sector, such as the rare earth elements used in wind turbines \citep{imholte2018assessment}.} Rents from the sector-specific resource appear in consumer income at the price $P_{Q}$. The fixed factor implies a diminishing marginal product of capital in $R$; as prime locations for wind or solar become scarce, the marginal site produces less energy per solar panel or wind turbine due to suboptimal conditions. After totally differentiating, the production function is given by:
\begin{align}
    \hat R &= \theta_{RK} \hat K_R, \label{Eq:ProductionR}
\end{align}
where $\theta_{RK} \equiv \frac{(1+t)P_KK_R}{(P_R+s)R}$ is the input cost share for capital (i.e., primary inputs). Recall that $t$ is the \textit{ad valorem} tax on the homogeneous capital. The revenue from the capital tax pays for the per unit carbon offset price, $s$, given to renewable energy production. Thus, $(P_R+s)R$ is the net of offset payment revenue in sector $R$. The CRTS assumption implies the zero-profit condition $(P_R +s) R = (1+t)P_K K_R + P_{Q}\bar{Q}$ and totally differentiating finds:
\begin{align}
    (1-\gamma)\hat P_R + \gamma\hat s + \hat R &= \theta_{RK} (\hat t + \hat P_K + \hat K_R) + \theta_{RQ} \hat{P}_{Q}, \label{Eq:ZeroProfitR}
\end{align}
where $\hat s \equiv \frac{ds}{s}$ is the change in the carbon offset price. The definition of $\hat s$ requires an initial, non-zero carbon offset, and then considers a small change in the offset value.\footnote{We note that $s$ could represent a composite of offset payments from several sources. For instance, it could be a combination of traditional carbon offset payments, the federal Clean Electricity Production Tax Credit, and state renewable energy credits.} We define the fixed factor input cost share as $\theta_{RQ} \equiv \frac{P_Q\bar Q}{(P_R+s)R}$, and then by CRTS it follows that $1  = \theta_{RK} + \theta_{RQ}$. Note that price, or rent, of the fixed factor can vary $(\hat P_Q)$. Also, we define $\gamma \equiv \frac{s}{s+P_R}$ as the share of the clean energy sector's return due to the offset.\footnote{Under perfect capital mobility and homogeneity, the tax applies to all sectors with capital.} Regardless of the fixed factor ($\hat Q = 0$), the substitution elasticity in $R$ is still defined and given by $\sigma_R\geq0$, and applying the definition recovers:
\begin{align}
    \hat K_R &= \sigma_R (\hat P_{Q} - \hat{t} - \hat P_K). \label{Eq:ElastR}
\end{align}
This $\sigma_R$ is not the renewable energy supply elasticity (although the two concepts are related). Rather, it measures how the clean energy sector substitutes between inputs.

For production of fossil fuel energy, we assume the production is simply $F = K_F$. That is, the production of $F$ only requires use of the primary input (capital), and totally differentiating leads to:
\begin{align}
    \hat F &= \hat K_F. \label{Eq:ProductionF}
\end{align}
This modeling choice follows from our national-level interpretation of this setting where the U.S. (or another other country or region) has taxing authority to raise revenue to buy offsets that subsidize local renewable energy production. In other words, U.S. tax revenue is subsidizing U.S. renewable energy projects. In the U.S., the number of feasible renewable energy sites is limited, leading to the fixed factor in the production of $R$. However, production of fossil fuel energy $F$ is not limited by fixed factors as the U.S. can import virtually unlimited quantities of fossil fuels. The prior literature has assumed fixed factors of production in the renewable energy sector as we do (see \citet{karney2025model}), or no fixed factors of production (see \citet{goulder2017confronting}).

The single input to $F$ implies no substitution in production However, the zero-profit condition, $P_F F = (1+t) P_K K_F$, due to CRTS still holds and totally differentiating reveals:
\begin{align}
    \hat P_F + \hat F &= \hat t + \hat P_K + \hat K_F. \label{Eq:ZeroProfitF}
\end{align}

The many identical households in our model face the budget constraint $P_X X \leq I = P_K \bar{K}+P_{Q}\bar{Q}+B$, where $I$ is total income (equivalent to GDP) and $B$ is the government's net budget given by $B=t(P_K\bar{K})-sR$. As Appendix \ref{App:LogLinear} shows, the zero-profit conditions imply the representative consumer's budget constraint and implies that fully voluntary offsets are not admissible in the model. To focus attention on the effects of carbon offsets, we assume a balanced budget constraint that all tax revenue is spent by the government on offset payments such that $B=0$. In equilibrium, consumers spend all their income on $X$. Also, since we fix the total amount of capital, the capital tax is equivalent to a lump-sum tax on consumers. Totally differentiating the government's balanced budget constraint finds:
\begin{align}
    \hat s + \hat R  = \psi \hat t + \hat P_K, \label{Eq:BalanceBudget}
\end{align}
where $\psi \equiv \frac{1+t}{t}$ scales the \textit{ad valorem} capital tax relative to the per unit carbon offset price.

Finally, we close the general equilibrium system by defining $X$ as the numeraire good and setting $P_X \equiv 1$. This numeraire choice is equivalent to setting the consumer price index (CPI) as the numeraire since there is only one consumption good. Totally differentiating $P_X$ simply yields:
\begin{align}
    \hat P_X &= 0 \label{Eq:Numeraire}.
\end{align}
When $\hat P_X=0$, the existing per-unit offset price $s>0$ is held constant in real terms relative to the CPI and this avoids potentially perverse outcomes as demonstrated by \citet{garnache2022environmental}. 

The linear system defined by equations \ref{Eq:Resource} through \ref{Eq:Numeraire} includes two policy variables $\{\hat t,\hat s\}$ along with 13 endogenous price-quantity variables $\{ \hat K_X,\hat K_R, \hat K_F, \hat X, \hat E, \hat F, \hat R, \hat P_K, \hat P_X, \hat P_E, \hat P_F, \hat P_R, \hat{P}_{Q} \}$. We study the effect of exogenously increasing the carbon offset payment $(\hat s>0)$ such that $\hat t$ becomes the last endogenous variable (via the balanced-budget condition). 

What about additional versus non-additional offsets? We start with the premise that the initial equilibrium contains non-additional $R$. That is, some proportion of $R$ is not due to the existing offset price, meaning some offsets are non-additional by traditional measures (see Section~\ref{Sec:Additionality} for further discussion). Assuming that additional and non-additional $R$ are perfect substitutes, total $R$ can be decomposed into $R = R^A + R^N$, where $R^A > 0$ denotes the additional component of $R$, and $R^N > 0$ denotes the non-additional component.

In Appendix~\ref{App:Non-Additional}, we detail an expansion of the model outlined above where $R^A$ and $R^N$ are each produced in separate production sub-sectors. Then, we show that $\hat{R} = \hat{R}^A = \hat{R}^N$ under the added assumption that the additional and non-additional sub-sectors have the same cost functions.\footnote{This baseline assumption of identical costs functions allows us to identify the effects in our model without unnecessarily complicating our general analytical results. Numerical models looking to analyze specific offset policy proposals can relax this assumption when calculating particular policy outcomes.}
In other words, production of both the additional and non-additional offsets change in equal proportion when the offset price changes. Totally differentiating the linear substitution condition leads to the auxiliary equation $\hat{R} = \phi \hat{R}^A + (1-\phi) \hat{R}^N$ that generally represents the separation between additional and non-additional offsets in the model, where $\phi \equiv R^A/R$ is the initial share of additional offsets.\footnote{Note that the condition $\hat{R} = \hat{R}^A = \hat{R}^N$ is sufficient for the auxiliary equation to hold. However, if the non-additional offsets use more productive fixed resources, $Q$, per unit of output, then the marginal product of capital is higher in the non-additional offset sub-sector, and thus the capital input cost share is lower than in the additional offset sub-sector; that is, $\theta_{AK} > \theta_{NK}$ while $\theta_{AQ} < \theta_{NQ}$ (see Appendix~\ref{App:Non-Additional} for notation). In this case, the non-additional offset sub-sector is more fixed resource intensive, all else equal. Then, $\hat{R}^A\neq\hat{R}^N$, but $\hat{R} = \phi \hat{R}^A + (1-\phi) \hat{R}^N$ holds.} 

Figure \ref{Fig:Model2} provides an illustration of the modeled economy. Each box represents a production sector as well as a box for the representative consumer. The arrows from the consumer to the production sectors denote the flow of the primary inputs $(K_F,K_R,K_X)$. The arrows return to the consumer show the quantities of the final good $(X)$ and pollution $(Z_F,Z_X)$ that enter the consumers utility function. (Note that $Z_F$ and $Z_X$ are not explicitly solved in the linear system as they are perfectly correlated with $F$ and $X$, respectively, as discussed in Section \ref{Sec:Additionality}.)  The figure illustrates that sector $E$ simply transforms clean and dirty energy into a composite energy service for the final goods sector. As we show in Section \ref{Sec:Numeric}, the elasticity of substitution in sector $E$ given by $\sigma_E$ is key to determining outcomes.


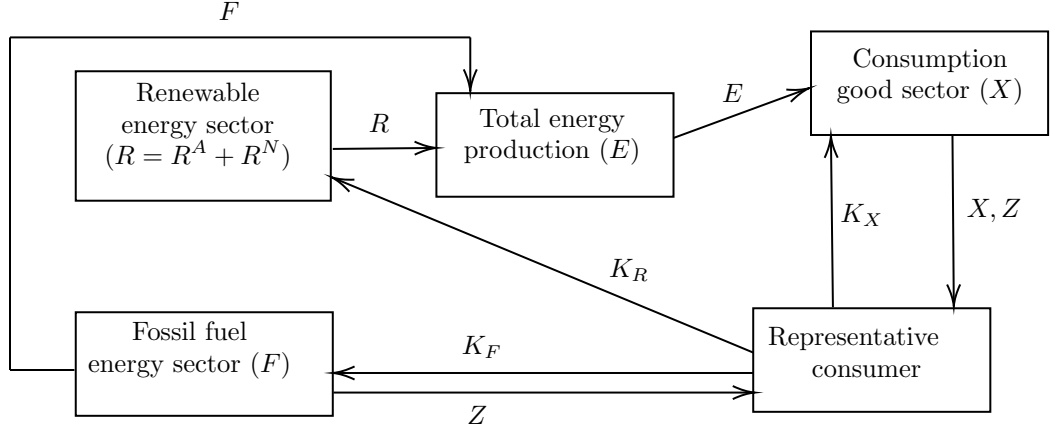
\begin{figure}

\tikzset{every picture/.style={line width=0.75pt}} 

\begin{tikzpicture}[x=0.75pt,y=0.75pt,yscale=-1,xscale=1]

\draw   (431,178) -- (550,178) -- (550,230) -- (431,230) -- cycle ;
\draw   (91,59) -- (219,59) -- (219,124) -- (91,124) -- cycle ;
\draw   (91,180) -- (220,180) -- (220,232) -- (91,232) -- cycle ;
\draw   (272,70) -- (391,70) -- (391,122) -- (272,122) -- cycle ;
\draw   (460,39) -- (579,39) -- (579,91) -- (460,91) -- cycle ;
\draw    (391,92.5) -- (457.12,68.19) ;
\draw [shift={(459,67.5)}, rotate = 159.81] [color={rgb, 255:red, 0; green, 0; blue, 0 }  ][line width=0.75]    (10.93,-3.29) .. controls (6.95,-1.4) and (3.31,-0.3) .. (0,0) .. controls (3.31,0.3) and (6.95,1.4) .. (10.93,3.29)   ;
\draw    (220,98) -- (270,97.52) ;
\draw [shift={(272,97.5)}, rotate = 179.45] [color={rgb, 255:red, 0; green, 0; blue, 0 }  ][line width=0.75]    (10.93,-3.29) .. controls (6.95,-1.4) and (3.31,-0.3) .. (0,0) .. controls (3.31,0.3) and (6.95,1.4) .. (10.93,3.29)   ;
\draw    (431,200.25) -- (221.35,113.02) ;
\draw [shift={(219.5,112.25)}, rotate = 22.59] [color={rgb, 255:red, 0; green, 0; blue, 0 }  ][line width=0.75]    (10.93,-3.29) .. controls (6.95,-1.4) and (3.31,-0.3) .. (0,0) .. controls (3.31,0.3) and (6.95,1.4) .. (10.93,3.29)   ;
\draw    (431,210.5) -- (222,210.5) ;
\draw [shift={(220,210.5)}, rotate = 360] [color={rgb, 255:red, 0; green, 0; blue, 0 }  ][line width=0.75]    (10.93,-3.29) .. controls (6.95,-1.4) and (3.31,-0.3) .. (0,0) .. controls (3.31,0.3) and (6.95,1.4) .. (10.93,3.29)   ;
\draw    (471,177.5) -- (470.02,93.5) ;
\draw [shift={(470,91.5)}, rotate = 89.33] [color={rgb, 255:red, 0; green, 0; blue, 0 }  ][line width=0.75]    (10.93,-3.29) .. controls (6.95,-1.4) and (3.31,-0.3) .. (0,0) .. controls (3.31,0.3) and (6.95,1.4) .. (10.93,3.29)   ;
\draw    (531,91) -- (531.65,175) ;
\draw [shift={(531.67,177)}, rotate = 269.56] [color={rgb, 255:red, 0; green, 0; blue, 0 }  ][line width=0.75]    (10.93,-3.29) .. controls (6.95,-1.4) and (3.31,-0.3) .. (0,0) .. controls (3.31,0.3) and (6.95,1.4) .. (10.93,3.29)   ;
\draw    (90.33,209) -- (58,209) ;
\draw    (58,209) -- (58,42) ;
\draw    (58,42) -- (289,42) ;
\draw    (289,42) -- (289,67) ;
\draw [shift={(289,69)}, rotate = 270] [color={rgb, 255:red, 0; green, 0; blue, 0 }  ][line width=0.75]    (10.93,-3.29) .. controls (6.95,-1.4) and (3.31,-0.3) .. (0,0) .. controls (3.31,0.3) and (6.95,1.4) .. (10.93,3.29)   ;
\draw    (220.33,220.33) -- (429.67,220.33) ;
\draw [shift={(431.67,220.33)}, rotate = 180] [color={rgb, 255:red, 0; green, 0; blue, 0 }  ][line width=0.75]    (10.93,-3.29) .. controls (6.95,-1.4) and (3.31,-0.3) .. (0,0) .. controls (3.31,0.3) and (6.95,1.4) .. (10.93,3.29)   ;

\draw (437,186) node [anchor=north west][inner sep=0.75pt]   [align=left] {\begin{minipage}[lt]{74.16pt}\setlength\topsep{0pt}
Representative 
\begin{center}
consumer
\end{center}

\end{minipage}};
\draw (101,63) node [anchor=north west][inner sep=0.75pt]   [align=left] {\begin{minipage}[lt]{74.72pt}\setlength\topsep{0pt}
\begin{center}
Renewable\\energy sector \\($\displaystyle R=R^{A} +R^{N}$)
\end{center}

\end{minipage}};
\draw (93,183) node [anchor=north west][inner sep=0.75pt]   [align=left] {\begin{minipage}[lt]{80.57pt}\setlength\topsep{0pt}
\begin{center}
Fossil fuel\\energy sector ($\displaystyle F$)
\end{center}

\end{minipage}};
\draw (284,76) node [anchor=north west][inner sep=0.75pt]   [align=left] {\begin{minipage}[lt]{67.04pt}\setlength\topsep{0pt}
\begin{center}
Total energy\\production ($\displaystyle E$)
\end{center}

\end{minipage}};
\draw (469,46) node [anchor=north west][inner sep=0.75pt]   [align=left] {\begin{minipage}[lt]{73.99pt}\setlength\topsep{0pt}
\begin{center}
Consumption\\good sector ($\displaystyle X$)
\end{center}

\end{minipage}};
\draw (414,64) node [anchor=north west][inner sep=0.75pt]   [align=left] {$\displaystyle E$};
\draw (236.67,79) node [anchor=north west][inner sep=0.75pt]   [align=left] {$\displaystyle R$};
\draw (161.33,22.67) node [anchor=north west][inner sep=0.75pt]   [align=left] {$\displaystyle F$};
\draw (536.67,121.67) node [anchor=north west][inner sep=0.75pt]   [align=left] {$\displaystyle X,Z$};
\draw (473.33,124) node [anchor=north west][inner sep=0.75pt]   [align=left] {$\displaystyle K_{X}$};
\draw (357.17,151.5) node [anchor=north west][inner sep=0.75pt]   [align=left] {$\displaystyle K_{R}$};
\draw (283.33,190.33) node [anchor=north west][inner sep=0.75pt]   [align=left] {$\displaystyle K_{F}$};
\draw (286,225.67) node [anchor=north west][inner sep=0.75pt]   [align=left] {$\displaystyle Z$};

\end{tikzpicture}
\caption{Illustration of the model}\label{Fig:Model2}
\end{figure}


\section{Additionality} \label{Sec:Additionality}

Carbon offset programs function under the assumption that offsets are both real and additional. In other words, it is assumed that the activity results in an actual reduction in emissions that would not have occurred in the absence of the offset payment. There are multiple methodologies to determine the impacts offset programs \citep{schneider2009assessing}. However, all methods are applied at the project-level, and tend to inadequately account for market spillover effects and other types of leakage. \cite{millard2010constructing} document the leakage assessment requirements of a variety of CDM methodologies, and note that virtually all energy-sector methodologies ignore market spillover effects. That is, the renewable energy produced is assumed to directly offset an equivalent amount of fossil fuel energy with no impacts on energy or other market prices and quantities. \cite{vohringer2006attribute} corroborates this, finding that most CDM projects do not account for leakage due to market and price effects. \cite{calel2025carbon} document this for CDM wind farm projects in India, noting:
\begin{displayquote}The [offset] revenue is estimated by multiplying the electricity that would be generated by a `business-as-usual' emissions factor. Most projects use a factor equal to the generation-weighted average carbon dioxide emissions per unit of net electricity generation from all generating power plants serving the same regional grid (tCO$_2$/MWh). The assumption is that the new project would replace an equivalent amount of `business-as-usual' generation capacity, avoiding the associated emissions.\end{displayquote}

\cite{schneider2009assessing} provides an overview of additionality assessment in the CDM. He finds that the two most prevalent additionality assessment methods are barrier analysis and investment analysis, both of which are project-level assessments focused on direct emissions reductions. Barrier analysis requires that the offset-producing project demonstrate that there is some sort of barrier preventing the project from moving forward without the offset payment. This barrier is often financial, but could be technological or policy-related. Investment analysis involves demonstrating that the offset-generating project is not the most financially attractive option. Then, these projects are assumed to directly offset an equivalent amount of dirty generation, where the dirty quantity reduced determines the offset credits that the projects earn. We refer to these commonly used categorizations of offsets by the above screening tools as ``additional (or non-additional) by traditional measures.'' These project-level assessments of additionality and crediting for direct emissions reductions align with our conventional carbon metric accounting as outlined below.


\subsection{Decomposition of Offset Types}\label{Sec:Non-additional}

Here, we discuss more subtle categorizations of offsets by creating a taxonomy to decompose the offset types (and different sources of emissions changes) in our general equilibrium model. This taxonomy makes a distinction between the \textit{levels} of additional and non-additional offsets in the initial equilibrium as compared to offset \textit{changes} due to policy changes. Our general equilibrium setting requires such a categorization because all quantities need to be accounted for in the model, even the non-additional offsets. That is, once a non-additional offset is identified in general equilibrium, it still needs to be tracked for carbon accounting purposes, unlike in the project-level perspective. The offsets in this model can then be decomposed into four types as follows:
\begin{enumerate}
    \item Inframarginal offsets that are non-additional by traditional measures: $R^N$;
    \item Extensive-margin offsets that are additional by traditional measures: $R^A$;
    \item Intensive-margin offsets that are additional by traditional measures: $dR^A$; and
    \item Intensive-margin offsets that are non-additional by traditional measures: $ d R^N$.
\end{enumerate}
Note that $d R^i = R^i \hat{R}^i$ for $i \in \{A,N \}$ is the absolute quantity change in sub-sector $R^i$'s output. In defining the offset types, we use the term ``inframarginal'' to refer to the subset of $R$ that would have been produced without implementing the initial offset policy. Then, ``extensive margin'' refers to offsets that would not exist without the initial offset policy, while ``intensive margin'' refers to offsets that occur in response to the increase in the offset price. These three terms identify the four groups of offsets in the model. Importantly, $dR^N$ is an intensive-margin offset despite coming from the non-additional sector because it only exists in response to the increase in offset price.\footnote{The use of ``extensive'' and ``intensive'' margin here differs from their use in labor economics, for example, while noting that CRTS production makes no distinction between the number of firms and output per firm.}

We discuss each offset type in turn. To start, much of the empirical literature focuses on identifying the quantity of Type 1 non-additional offsets (see, for example, \citet{calel2025carbon}). In our taxonomy, $R^N$ is inframarginal as it neither arises due to the carbon offset's introduction $(s)$ nor the change in the offset price~$(\hat{s})$. Rather, $R^N$ is a parameter in the model. Policymakers care a great deal about this type because the sub-sector $R^N$ receives payments that it does not need to receive to exist.

Next, we call $R^A$ the Type 2 extensive-margin additional offsets. Again, the empirical literature focuses on identifying and measuring $R^A$ (similar to the goal of identifying $R^N$). However, offset payments to this sub-sector are typically deemed useful expenditures since $R^A$ would not exist without the payments. We call this an extensive margin, as the counterfactual without the offset payment involves no additional offsets in the initial equilibrium $(i.e., R^A=0).$

Continuing, Type 3 intensive-margin additional offsets represent the change in initially additional offsets. It is natural for the change in additional offsets $(dR^A)$ to be deemed additional by traditional measures (and this is particularly salient in partial equilibrium). The change in additional offsets can only exists if initially additional offsets exists and thus the designation as ``intensive'' margin.

Finally, to the best of our knowledge, Type 4 offsets have not been previously identified in the literature. These intensive-margin offsets occur because the firms in the $R^N$ sub-sector still profit maximize on the margin. That is, the change in the offset price provides a change in the marginal incentive for $R^N$ firms (although the initial quantity of $R^N$ would have been provided without the initial offset price $s$, see Type 1 offsets). In the wind generation context, an increase in the offset price induces firms to run their wind turbines more often, perhaps by minimizing down time.\footnote{\citep{johnson2026negative} find that the federal U.S. Production Tax Credit for eligible wind generators caused more negatively priced electricity hours, demonstrating the existence of these marginal effects.}.

Consider the extreme case where $R \approx  R^N$ and $R^A \approx 0$. Despite no traditionally additional offsets, policymakers would induce more intensive-margin offsets $(dR^N)$ in response to an increase in the offset price. Thus, even when all initial offsets are non-additional by traditional measures, there exists an avenue through which emissions may decrease via an increase in the offset price and subsequent quantity. This highlights an understudied channel whereby offset payments induce an intensive-margin adjustment from traditionally non-additional offset producers, suggesting that traditional measures of non-additionality may understate additionality in some cases.

In summary, $dR=dR^A+dR^N$ is always the quantity change in intensive margin additional offsets regardless of the share of additional and non-additional offsets in the initial equilibrium. Therefore, without loss of generality, Section~\ref{Sec:Analytical} solves for $\hat{R}$ to show how a change in the offset price affects the quantity of offsets and subsequent change in emissions.


\subsection{Carbon Accounting Metrics}\label{Sec:Definitions}

Here, we define two types of carbon accounting metrics to explore their general equilibrium implications: conventional and aggregate carbon accounting. Conventional carbon accounting is designed to capture the emissions changes estimated via standard carbon accounting practices, as outlined at the start of Section~\ref{Sec:Additionality}. Conversely, aggregate carbon accounting is designed to capture aggregate greenhouse gas emissions changes. Importantly, both are carbon accounting metrics, rather than offset types (see definitions above).

Conventional carbon accounting measures the increase in quantity of additional renewable energy $(dR^A>0)$ induced by the higher offset price (see Type 3 offsets) and then credits those renewable energy projects with the corresponding level of emissions that would have been produced by the same quantity of energy under fossil fuel energy production (the direct emissions change). We call this estimate the ``conventional carbon accounting'' metric for offsets, denoted~$\Omega_{C}$. It is defined as $\Omega_C \equiv -\xi_F dR^A$ and equivalent to:
\begin{align}
      \Omega_C = -(\xi_F R^A) \hat{R}^A = -(\xi_F \phi R) \hat{R}, \label{Eq:OffsetConventional}
\end{align}
where $\xi_F$ is the emission factor of fossil fuel energy $(F)$ such that $\xi_FR^A$ (or $\xi_F \phi R$) is the implied total conventional emissions offset prior to the change in the offset price (recalling $\phi \equiv R^A/R$ and using the result that $\hat{R}^A=\hat{R}$, see Section~\ref{Sec:Model}). Our definition of conditional carbon accounting means $\Omega_C<0$ when $\hat{R}>0$. Therefore, a larger increase in renewable energy $(\hat{R})$, a larger initial share of additional offsets $(\phi)$, or a larger emission factor for the fossil fuel sector $(\xi_F)$ imply a higher conventional measure, and hence a larger amount of credited emission reductions, all else equal. 

Yet, the purported purpose of carbon offset programs is to reduce carbon emissions. If so, then the true interest of policymakers from an emissions-only standpoint is the change in total emissions induced by the increase in the offset price. Total carbon emissions are given by $Z = Z_F +Z_X$ where $Z_F$ denotes emissions from the fossil fuel sector and $Z_X$ represents emissions in the rest of the economy. Thus, the production of fossil fuel energy $(F)$ and the production of the final good $(X)$ both generate emissions but the renewable energy sector $(R)$ does not. Then, following \cite{fullerton2025determines}, we define $Z_F\ \equiv \xi_FF$ and $Z_X \equiv \xi_XK_X$ where $\xi_X$ is the emission factor in the final goods sector related to the primary input $K_X$.\footnote{Note that emissions could be carbon or other greenhouse gases such as methane, nitrous oxides, or HFCs. Thus, the emissions intensity terms, $\xi_F$ and $\xi_X$, can capture both differing emissions rates produced by the two sectors but also different global warming potentials of the greenhouse gas emissions.} Substituting implies $Z=\xi_FF+\xi_XK_X$ and totally differentiating $Z$ then finds the proportional change in total emissions as follows:
\begin{align}
    \hat{Z}  = \rho_F \hat F + \rho_X \hat K_X. \label{Eq:Zhat}
\end{align}
where $\rho_F \equiv \frac{\xi_F F}{Z}$ and $\rho_X \equiv \frac{\xi_X K_X}{Z}$ are the initial shares of carbon emissions in the economy such that $1=\rho_F+\rho_X$ (see Appendix \ref{App:LogLinear} for derivations). As shown in Section \ref{Sec:Analytical}, the signs of $\hat F$ and $\hat K_X$ are generally ambiguous, and thus the sign of $\hat Z$ is generally ambiguous.

We define the quantity change in total emissions as the ``aggregate carbon accounting'' metric of offsets, denoted $\Omega_A$. Formally, it is defined as $\Omega_A \equiv \xi_FdF+\xi_XdK_X$ and equivalent to:
\begin{align}
    \Omega_A = (\xi_FF)\hat F+(\xi_XK_X)\hat K_X. \label{Eq:OffsetAggregate}
\end{align}    
Note that $\Omega_A$ can be negative or positive (similar to $\hat Z$).\footnote{Also, $\Omega_A = dZ$, which implies that equations \ref{Eq:Zhat} and \ref{Eq:OffsetAggregate} are related as follows: $\hat{Z}=\frac{\Omega_A}{Z} \Leftrightarrow Z\hat{Z}=\Omega_A.$} 

In general, it could be that the metric of conventional or aggregate carbon accounting is larger for a given increase in the offset price. We measure this difference via the ratio:
\begin{align}
    \Delta & \equiv \frac{\Omega_A-\Omega_C}{\Omega_C} = \frac{\Omega_A}{\Omega_C}-1, \label{Eq:Delta}\\
    &= - \left[ \frac{(\xi_FF)\hat F + (\xi_XK_X)\hat K_X}{(\xi_F \phi R) \hat{R}} \right] -1 . \nonumber
\end{align}

If $\Omega_C =\Omega_A$ such that conventional carbon accounting perfectly measures aggregate emissions changes, then $\Delta=0$. For many parameter values (see Section~\ref{Sec:Numeric}), we find $\Delta<0$ implying the conventional carbon accounting metric overstates emissions reductions relative to the aggregate metric. A value of $\Delta=-0.20$ means conventional carbon accounting overstates true emissions reductions by 20 percent while $\Delta=-1$ means there are no aggregate emission reductions $(\Omega_A=0)$ and the conventional carbon accounting metric records emissions reductions where none exist.\footnote{The ratio $\Delta$ is not defined when $\Omega_C=0$, but $\Omega_C>0$ when $\hat R>0$ and $\phi > 0$, and this always occurs with an increase in the offset price as shown in Section \ref{Sec:Analytical}.} In the extreme, $\Delta<-1$ means $\hat Z>0$ (implying $\Omega_A>0$) and this case is known as ``backfire'' where total carbon emissions actually increase as result of an offset policy. It is also possible to find $\Delta>0$ which means that conventional carbon accounting understates the true emission reductions in general equilibrium. Interestingly, assuming $\Omega_C<0$ and $\Omega_A<0$, then $\phi \rightarrow0$ implies $\Delta\rightarrow\infty$, and this demonstrates how ignoring the intensive-margin offsets from initially non-additional offsets (see Type 4 offsets) under-counts real emission reductions. We use the ratio in \ref{Eq:Delta} to compare the two carbon accounting metrics across scenarios with different levels of conventional carbon accounting. Using the closed-form solutions for the endogenous variables reported in Section \ref{Sec:Analytical}, we examine the extent to which these two carbon accounting metrics differ via the measure $\Delta$, and then discuss implications for policymakers and carbon accounting.

Figure~\ref{fig:taxonomy} provides an overview of the relationships between offset types and carbon accounting metrics defined in this section. To start, the blue boxes denote emission sources. These boxes include the four offset types and the general equilibrium effects that lead to emission changes in other sectors of the model's economy. Then, the green boxes denote how the emission sources are measured. Traditional additionality screening tools measure the Type 1 and 2 offsets at the initial equilibrium (and also the parameter $\phi$). A change in the offset prices leads to Type 3 and 4 offsets as well as the general equilibrium changes. However, the conventional carbon accounting metric $(\Omega_C)$ only captures Type 3 offsets, while all of the offsets induced by the price change contribute to the aggregate carbon accounting metric $(\Omega_A)$.


\begin{figure}
    
\resizebox{\textwidth}{!}{

\tikzset{every picture/.style={line width=0.75pt}} 

\begin{tikzpicture}[x=0.75pt,y=0.75pt,yscale=-1,xscale=1]

\draw  [fill={rgb, 255:red, 213; green, 229; blue, 247 }  ,fill opacity=1 ] (70.5,23) .. controls (70.5,16.37) and (75.87,11) .. (82.5,11) -- (174.5,11) .. controls (181.13,11) and (186.5,16.37) .. (186.5,23) -- (186.5,59) .. controls (186.5,65.63) and (181.13,71) .. (174.5,71) -- (82.5,71) .. controls (75.87,71) and (70.5,65.63) .. (70.5,59) -- cycle ;
\draw  [fill={rgb, 255:red, 213; green, 229; blue, 247 }  ,fill opacity=1 ] (71.5,231) .. controls (71.5,224.37) and (76.87,219) .. (83.5,219) -- (175.5,219) .. controls (182.13,219) and (187.5,224.37) .. (187.5,231) -- (187.5,267) .. controls (187.5,273.63) and (182.13,279) .. (175.5,279) -- (83.5,279) .. controls (76.87,279) and (71.5,273.63) .. (71.5,267) -- cycle ;
\draw  [fill={rgb, 255:red, 213; green, 229; blue, 247 }  ,fill opacity=1 ] (70.5,92) .. controls (70.5,85.37) and (75.87,80) .. (82.5,80) -- (174.5,80) .. controls (181.13,80) and (186.5,85.37) .. (186.5,92) -- (186.5,128) .. controls (186.5,134.63) and (181.13,140) .. (174.5,140) -- (82.5,140) .. controls (75.87,140) and (70.5,134.63) .. (70.5,128) -- cycle ;
\draw  [fill={rgb, 255:red, 213; green, 229; blue, 247 }  ,fill opacity=1 ] (70.5,161) .. controls (70.5,154.37) and (75.87,149) .. (82.5,149) -- (174.5,149) .. controls (181.13,149) and (186.5,154.37) .. (186.5,161) -- (186.5,197) .. controls (186.5,203.63) and (181.13,209) .. (174.5,209) -- (82.5,209) .. controls (75.87,209) and (70.5,203.63) .. (70.5,197) -- cycle ;
\draw  [fill={rgb, 255:red, 213; green, 229; blue, 247 }  ,fill opacity=1 ] (70.5,302.8) .. controls (70.5,295.18) and (76.68,289) .. (84.3,289) -- (172.7,289) .. controls (180.32,289) and (186.5,295.18) .. (186.5,302.8) -- (186.5,344.2) .. controls (186.5,351.82) and (180.32,358) .. (172.7,358) -- (84.3,358) .. controls (76.68,358) and (70.5,351.82) .. (70.5,344.2) -- cycle ;
\draw   (64,10.5) .. controls (59.33,10.47) and (56.98,12.78) .. (56.95,17.45) -- (56.58,65.7) .. controls (56.53,72.37) and (54.17,75.68) .. (49.5,75.64) .. controls (54.17,75.68) and (56.47,79.03) .. (56.42,85.7)(56.45,82.7) -- (56.05,133.95) .. controls (56.02,138.62) and (58.33,140.97) .. (63,141) ;
\draw   (64,149.5) .. controls (59.33,149.5) and (57,151.83) .. (57,156.5) -- (57,243) .. controls (57,249.67) and (54.67,253) .. (50,253) .. controls (54.67,253) and (57,256.33) .. (57,263)(57,260) -- (57,349.5) .. controls (57,354.17) and (59.33,356.5) .. (64,356.5) ;
\draw   (195,139.5) .. controls (199.67,139.5) and (202,137.17) .. (202,132.5) -- (202,85.5) .. controls (202,78.83) and (204.33,75.5) .. (209,75.5) .. controls (204.33,75.5) and (202,72.17) .. (202,65.5)(202,68.5) -- (202,18.5) .. controls (202,13.83) and (199.67,11.5) .. (195,11.5) ;
\draw  [fill={rgb, 255:red, 232; green, 249; blue, 227 }  ,fill opacity=1 ] (222,53.9) .. controls (222,46.22) and (228.22,40) .. (235.9,40) -- (381.1,40) .. controls (388.78,40) and (395,46.22) .. (395,53.9) -- (395,95.6) .. controls (395,103.28) and (388.78,109.5) .. (381.1,109.5) -- (235.9,109.5) .. controls (228.22,109.5) and (222,103.28) .. (222,95.6) -- cycle ;
\draw   (196,208.5) .. controls (200.67,208.42) and (202.96,206.05) .. (202.88,201.38) -- (202.67,188.38) .. controls (202.56,181.71) and (204.83,178.34) .. (209.5,178.27) .. controls (204.83,178.34) and (202.45,175.05) .. (202.34,168.39)(202.39,171.38) -- (202.12,155.38) .. controls (202.04,150.71) and (199.67,148.42) .. (195.01,148.5) ;
\draw  [fill={rgb, 255:red, 232; green, 249; blue, 227 }  ,fill opacity=1 ] (222,157.9) .. controls (222,149.95) and (228.45,143.5) .. (236.4,143.5) -- (387.6,143.5) .. controls (395.55,143.5) and (402,149.95) .. (402,157.9) -- (402,201.1) .. controls (402,209.05) and (395.55,215.5) .. (387.6,215.5) -- (236.4,215.5) .. controls (228.45,215.5) and (222,209.05) .. (222,201.1) -- cycle ;
\draw   (415,359.5) .. controls (419.67,359.5) and (422,357.17) .. (422,352.5) -- (422,265.5) .. controls (422,258.83) and (424.33,255.5) .. (429,255.5) .. controls (424.33,255.5) and (422,252.17) .. (422,245.5)(422,248.5) -- (422,158.5) .. controls (422,153.83) and (419.67,151.5) .. (415,151.5) ;
\draw  [fill={rgb, 255:red, 232; green, 249; blue, 227 }  ,fill opacity=1 ] (440,233.1) .. controls (440,224.21) and (447.21,217) .. (456.1,217) -- (580.9,217) .. controls (589.79,217) and (597,224.21) .. (597,233.1) -- (597,281.4) .. controls (597,290.29) and (589.79,297.5) .. (580.9,297.5) -- (456.1,297.5) .. controls (447.21,297.5) and (440,290.29) .. (440,281.4) -- cycle ;
\draw  [color={rgb, 255:red, 0; green, 0; blue, 0 }  ,draw opacity=1 ][fill={rgb, 255:red, 230; green, 230; blue, 230 }  ,fill opacity=1 ] (59,-45.5) -- (199,-45.5) -- (199,-2) -- (59,-2) -- cycle ;
\draw  [fill={rgb, 255:red, 230; green, 230; blue, 230 }  ,fill opacity=1 ] (219,-43.5) -- (607,-43.5) -- (607,-2) -- (219,-2) -- cycle ;

\draw (102,22) node [anchor=north west][inner sep=0.75pt]   [align=left] {\begin{minipage}[lt]{36.18pt}\setlength\topsep{0pt}
Type 1,
\begin{center}
$\displaystyle R^{N}$
\end{center}

\end{minipage}};
\draw (102,93) node [anchor=north west][inner sep=0.75pt]   [align=left] {\begin{minipage}[lt]{36.18pt}\setlength\topsep{0pt}
\begin{center}
Type 2,\\$\displaystyle R^{A}$
\end{center}

\end{minipage}};
\draw (102,163) node [anchor=north west][inner sep=0.75pt]   [align=left] {\begin{minipage}[lt]{36.18pt}\setlength\topsep{0pt}
\begin{center}
Type 3,\\$\displaystyle dR^{A}$
\end{center}

\end{minipage}};
\draw (102,231) node [anchor=north west][inner sep=0.75pt]   [align=left] {\begin{minipage}[lt]{36.18pt}\setlength\topsep{0pt}
\begin{center}
Type 4,\\$\displaystyle dR^{N}$
\end{center}

\end{minipage}};
\draw (93,295) node [anchor=north west][inner sep=0.75pt]   [align=left] {\begin{minipage}[lt]{52.05pt}\setlength\topsep{0pt}
\begin{center}
General \\equilibrium\\effects
\end{center}

\end{minipage}};
\draw (21.67,132.52) node [anchor=north west][inner sep=0.75pt]  [rotate=-269.83] [align=left] {Initial equilibrium};
\draw (4.67,307.97) node [anchor=north west][inner sep=0.75pt]  [rotate=-270.23] [align=left] {\begin{minipage}[lt]{80.23pt}\setlength\topsep{0pt}
\begin{center}
Induced by offset\\price change
\end{center}

\end{minipage}};
\draw (69,-33) node [anchor=north west][inner sep=0.75pt]   [align=left] {\begin{minipage}[lt]{82.65pt}\setlength\topsep{0pt}
\begin{center}
Emissions source
\end{center}

\end{minipage}};
\draw (360,-31) node [anchor=north west][inner sep=0.75pt]   [align=left] {Measured by};
\draw (232,54) node [anchor=north west][inner sep=0.75pt]   [align=left] {\begin{minipage}[lt]{106.68pt}\setlength\topsep{0pt}
\begin{center}
Traditional additionality\\screening tools
\end{center}

\end{minipage}};
\draw (247,151) node [anchor=north west][inner sep=0.75pt]   [align=left] {\begin{minipage}[lt]{88.92pt}\setlength\topsep{0pt}
\begin{center}
Conventional\\carbon accounting,\\$\displaystyle ( \Omega _{C})$
\end{center}

\end{minipage}};
\draw (456,228) node [anchor=north west][inner sep=0.75pt]   [align=left] {\begin{minipage}[lt]{88.92pt}\setlength\topsep{0pt}
\begin{center}
Aggregate\\carbon accounting,\\$\displaystyle ( \Omega _{A}$)
\end{center}

\end{minipage}};

\end{tikzpicture}
}
\caption{Relationships between offset types and carbon accounting metrics.}\label{fig:taxonomy}
\end{figure}
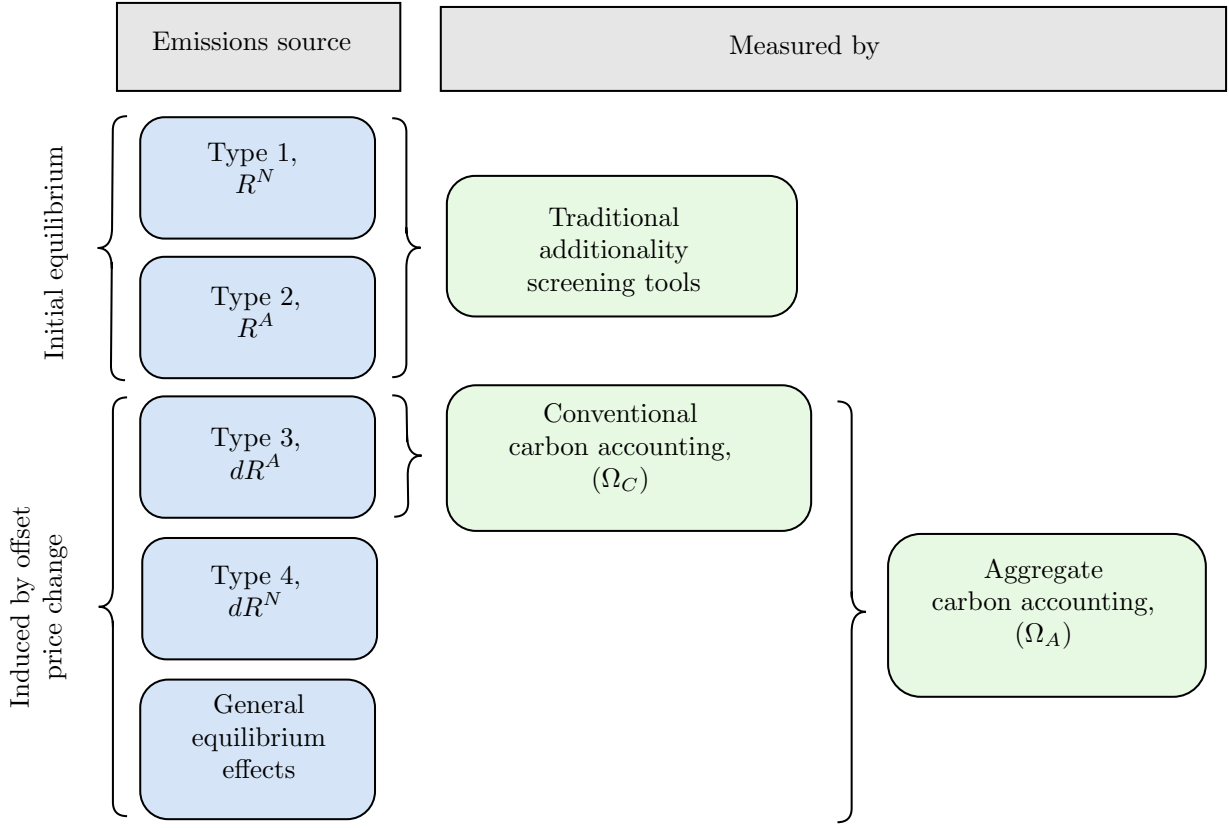


\section{Welfare} \label{Sec:Welfare}

\subsection{Change in Welfare}

In this section, we derive a welfare formula that shows the trade-off between utility from consumption and damages from emissions. With a large number of households, $n$, each household does not consider their own impact on aggregate emissions ($nZ$), and we assume that aggregate emissions enter separably into utility. The household receives direct utility from the consumer of good $X$ with the the representative consumer's utility function given by $U = U(X; nZ)$. We assume that $U$ is continuous, quasiconcave, and twice-differentiable such that $U$ is strictly increasing in $X$ and strictly decreasing in $nZ$. Aggregate emissions are set apart in the utility function with a semicolon to indicate that the level of aggregate emissions is not a choice variable; that is, households cannot optimize over $nZ$.

As shown in Appendix \ref{App:Welfare}, the change in welfare from an exogenous policy change is found to be:
\begin{align}
    \frac{dU}{\lambda I} = \underbrace{\hat X}_{\text{Marginal Cost}} - \underbrace{\mu \left( \frac{nZ}{I} \right) \hat Z}_{\text{Marginal Benefit}}, \label{Eq:Welfare}
\end{align}
where $\lambda$ is the marginal utility of income (and recalling that $I$ is total income, or GDP), and thus the left-hand side of equation \ref{Eq:Welfare} is the percent change in welfare using the money metric \citep{bovenberg1994environmental}. We define $\mu \equiv \frac{-1}{\lambda}\frac{\partial U}{\partial Z}>0$ as the marginal environmental damage (MED) from a change in $Z$. The MED is positive since a binding budget constraint implies $\lambda>0$ and higher pollution reduces utility ($\frac{\partial U}{\partial Z}<0$).

In a standard model where a negative externality is controlled by an existing tax, $\tau$, the welfare equation differs from our equation \ref{Eq:Welfare} in two important ways. First, final consumption goods $(e.g., \hat X)$ do not appear in the welfare equation as the existing tax internalizes the cost. Second, welfare terms with respect to pollution changes take the form $(\tau-\mu)\frac{nZ}{I}\hat Z$ such that a decrease in pollution $\hat Z<0$ increases welfare when the tax is set lower than the MED (meaning the initial $\tau$ may not be set optimally). As a result, welfare can increase or decrease depending on whether the pollutant is over- or under-regulated, and welfare is maximized when $\tau=\mu$ (i.e., a Pigouvian tax). That is, $\tau$ is the marginal cost of reducing emissions while $\mu$ is the marginal benefit, and maximizing welfare requires setting the marginal benefit equal to the marginal cost.

Our model does not have an explicit cost to producers for pollution, so a fixed parameter similar to $\tau$ does not appear in equation \ref{Eq:Welfare}. However, the final consumption good $\hat X$ appears in the welfare equation, and decreasing final good consumption is the indirect cost of subsidizing renewable energy production via offsets (and below we show $\hat X<0 $, always). This means $\hat X$ is the marginal cost of pollution reduction as labeled in equation \ref{Eq:Welfare}, and this marginal cost is not fixed, but a function of economic parameters such as the factor shares and production elasticities, and arises due to the distortion from changing relative prices induced by the offset payments. On the other hand, $\mu \left( \frac{nZ}{I} \right) \hat Z$ is the marginal benefit of carbon reductions. The marginal benefit can be negative in cases of backfire, $\hat{Z}>0$, when subsidizing clean energy increases total emissions.

Finally, since $\mu$ is a free parameter in this model, then for any $\hat Z<0$ there exists a $\mu^*$ such that the welfare gain is non-negative meaning the marginal benefit of subsidizing clean energy is greater than or equal to the marginal cost. Formally, when $\hat{Z}<0$, define $\mu^*$ as follows:
\begin{align}
    \mu^* \equiv  \left( \frac{I}{nZ} \right) \frac{\hat X}{\hat Z}, \label{Eq:MuStar}
\end{align}
such that $\mu^*$ is increasing in income $(I)$ as well as the loss of final good consumption $(\hat X)$, but decreasing in the initial amount of emissions $(Z)$ and aggregate carbon emission reductions $(\hat Z)$. As the closed-form solutions in Section \ref{Sec:Analytical} make clear, $\mu^*$ is not a function of policy change since $\hat s>0$ cancels out from the numerator and denominator of equation \ref{Eq:MuStar}; rather, $\mu^*$ is a function of the economic parameters and initial policy conditions $(s,t)$. Also, $\mu^*$ is not defined when $\hat{Z}>0$.

\subsection{MVPF}\label{Sec:MVPF}

The marginal value of public funds (MVPF) is a benefit-cost ratio for comparing outcomes across policies, and it is defined as the ratio of the benefits to individuals divided by the net cost to government \citep{bastani2024marginal}. A larger MVPF means that a spending policy has more ``bang-for-buck''. In the context of public goods, \cite{hahn2024welfare} provides a framework for evaluating the MVPF while emphasizing the importance of externalities as well as (potential) rebound and general equilibrium effects. Our model fully captures both the rebound and general equilibrium effects from a subsidy-based offset policy. In this section, we translate \cite{hahn2024welfare}'s framework to our setting to evaluate how the MVPF relates to the carbon accounting metrics defined in Section \ref{Sec:Additionality}. 

The MVPF's numerator is the willingness to pay (WTP) for a small increase in the offset price $(ds)$:
\begin{align*}
    WTP =\underbrace{Rds}_{\text{Offset Value}}+\underbrace{VdR}_{\text{Externality Value}},
\end{align*}
where $V$ is the marginal externality value per unit of $R$. The WTP is then comprised of the dollar value of the offset value and the dollar value of the externality. Rewriting the above equation in our ``hat'' notation leads to: $WTP =(sR)\hat{s}+(VR)\hat{R}$. The value of $V$ must account for all environmental benefits from a change in $R$ and not just the direct benefits from offsetting $F$. In other words, $V$ is related to the aggregate carbon accounting metric, defined as $\Omega_A\equiv\xi_FdF+\xi_XdK_X$. As shown below, $dF$ and $dK_X$ can be rewritten in terms of $dR$, so the aggregate carbon accounting metric can be redefined as $\Omega_A=-\zeta dR$, where $\zeta$ is a collection of parameters that includes the emission factors $(\xi_F,\xi_X)$ as well as elasticity and share parameters.\footnote{Alternatively, one can calculate $\zeta$ as follows: $\zeta = \frac{Z\hat{Z}}{R\hat{R}}=\frac{dZ}{dR}$.}
Then, the dollar value per unit for a change in $R$ is given by $V=\mu\zeta$, and since $\Omega_A$ is negative when emissions reductions occur then $\zeta>0$ (recalling $\mu>0$ is the MED). Substituting then finds:
\begin{align}
    WTP =(sR)\hat{s}+(\mu\zeta R)\hat{R}, \label{Eq:WTP}
\end{align}
and the closed-form expression for $\hat{R}$ (found in Section \ref{Sec:Analytical}) can be substituted into \ref{Eq:WTP} to recover WTP as a function of the exogenous policy variable $\hat{s}$ only.

The MVPF's denominator is the net cost to the government given by $Cost=Rds+sdR$, where $dR$ captures all rebound and general equilibrium effects. Again, switching to our hat notation yields:
\begin{align}
    Cost=sR(\hat{s}+\hat{R}).
\end{align}
Therefore, the MVPF in our setting is found to be:
\begin{align}
    MVPF=\frac{WTP}{Cost}=\frac{(sR)\hat{s}+(\mu\zeta R)\hat{R}}{sR(\hat{s}+\hat{R})}=\frac{1+(\mu\zeta/s)\epsilon}{1+\epsilon},\label{Soln:MVPF}
\end{align}
where $\epsilon=\hat{R}/\hat{s}$ is the general equilibrium elasticity of  renewable energy $(R)$ with respect to the offset price~$(s)$. The parameter value $\epsilon$ is implicitly defined in equation \ref{Soln:R} below. We find $\epsilon>0$, which means that an increase in the offset price always increases $R$ and implies that $\textrm{MVPF}>0$. Then, a policy is social welfare improving when $\textrm{MVPF}>1$ since the willingness to pay exceeds the cost on the margin. The necessary and sufficient condition for $\textrm{MVPF}>1$ is $|\mu\zeta|>s$. This condition means that the dollar value of the total externality benefit per unit of $R$ is greater than the per unit offset price. The MVPF is increasing in $\mu$ and $\zeta$; that is, as the environmental damage increases, or the aggregate emission offset increases, then the value of the offset program increases. Conversely, the MVPF decreases as the initial offset price increases, all else equal.

\subsection{Optimal Policy}\label{Sec:Optimal}

The model in Section \ref{Sec:Model} solves for the change in prices and quantities from one equilibrium to another from a change in the prevailing offset price. What is the optimal policy? Regardless of the complexities induces by an explicit clean energy sector with a fixed-factor of production, an optimal policy is straightforward: set a Pigouvian tax on emissions, $(t_Z)$, equal to the marginal environmental damage. That is, set $t_Z=-n\mu>0$ where $t_Z$ is a per unit tax on $Z=Z_F+Z_X$. As a corollary, the single instrument of a subsidy on $R$ alone cannot induce the socially optimal outcome. We check whether a subsidy on $R$ paired with an output tax recovers the first-best solution (in other words, we search for an optimal two-part instrument like in \cite{fullerton1997case}). We find that a two-part instrument consisting of a subsidy on $R$ and an output tax on $X$ cannot achieve the socially optimal outcome except in the special case where $\xi_F=\xi_X$. Only under the special case where $\xi_X=0$ (no externality in the final good sector; emissions only come from $F$), do we find that a two-part instrument consisting of a subsidy on $R$ and a tax on $E$ can achieve first best. See Supplemental (Online) Appendix for proofs of these results.
 

\section{Analytical Solutions} \label{Sec:Analytical}

Solving the model detailed in Section \ref{Sec:Model} yields the following solutions for $\hat R$ and $\hat K_R$ (see Supplemental Appendix for details):
\begin{align}
    \hat{R}{} = \theta_{RK}\hat{K}_R = \left[ \frac{\gamma\theta_{RK}}{\left( \theta_{RQ}/\sigma_R \right) + \left( AB/C\right) } \right] \hat{s} = \left( \frac{\gamma\theta_{RK}}{D} \right) \hat{s} = \epsilon\hat{s}, \label{Soln:R}
\end{align}
where $A$ and $B$ are a positive collection of share parameters defined by $A\equiv (1-\gamma(1-\theta_{XR}))>0$ and $B \equiv \left[ \theta_{RK}+\alpha_R\theta_{RQ} \right]>0$, respectively. Define $C \equiv \left[ \alpha_X\theta_{ER}\sigma_X+(\alpha_F+\alpha_X\theta_{ER})\sigma_E \right] \geq0$ as a non-negative collection of share and elasticity parameters. Finally, let $D$ denote the entire denominator to simplify notation. All parameters and collections are non-negative. Thus, $\hat R > 0$ when $\hat{s}>0$. In words, increasing the offset price paid to the clean sector increases the quantity of $R$. The parameter $\epsilon$ (in the final equality) is employed to evaluate the MVPF (as defined in equation \ref{Soln:MVPF} above).\footnote{Note that $\hat{R}$ is only defined when (i) $\sigma_R\neq0$ and (ii) both $\sigma_X\neq0$ and $\sigma_E\neq0$. If either condition does not hold, then $D$ is undefined and thus $\hat{R}$ is undefined. The condition $\sigma_R=0$ cannot hold because that implies $K_R$ and $\bar{Q}$ are perfect complements, but $\bar{Q}$ is fixed and thus $R$ would be fixed too, so the change in $R$ is undefined. The condition $\sigma_X=\sigma_E=0$ cannot hold because the rest of the economy has no flexibility to adjust to changes in $R$ via substitution (recalling the representative consumer only gets consumption utility from $X$).}

The elasticities of substitution play a key role in determining the value of $\hat R$. As the elasticities become small, $\hat R$ falls. Specifically, $\lim_{\sigma_i\to0}\ \hat{R}\to0 \ \forall i \in \{E,F,X\}$. Intuitively, if the production sectors cannot adjust in response to a change in the clean energy incentive, then the amount of clean energy cannot change. However, the converse does not hold as $\hat R$ is finite when any particular production sector tends towards perfectly substitutable inputs. In Section \ref{Sec:Numeric}, we numerically calculate $\hat R$ and other endogenous variables as a function of the elasticities.  

The signs of comparative statics for $\hat R$ with respect to the elasticities of substitution are all positive: $\frac{\partial \hat R}{\partial \sigma_R},\frac{\partial \hat R}{\partial \sigma_X},\frac{\partial \hat R}{\partial \sigma_E}>0$. These comparative statics are interpreted as follows. To start, a more elastic substitution in the production of clean energy $(\sigma_R \uparrow)$ means that production can shift away from the fixed factor, and attract capital, to increase renewable energy output. Next, a more flexible energy aggregation $(\sigma_E \uparrow)$ allows substitution towards clean energy and away from dirty energy without losing overall energy output. Finally, a more flexible final goods production sector $(\sigma_X \uparrow)$ means that production can utilize the extra energy. Therefore, more flexible production functions lead to larger increases in $R$, all else equal.

Recall that, under the assumption that the additional and non-additional sub-sectors have the same cost function, it is the case that $\hat{R} = \hat{R}^A=\hat{R}^N$ (see Section~\ref{Sec:Model} and Appendix~\ref{App:Non-Additional} for details). That is, the two sub-sectors have the same proportional change when the offset price changes. Thus, it follows that $\hat{R}^A$ and $\hat{R}^N$ have the same functional form outlined in equation~\ref{Soln:R}, and follow the same comparative statics.

Next, the change in fossil energy production is given by:
\begin{align}
    \hat F = \hat K_F = \left[ 1-\sigma_E \left( \frac{B}{\theta_{RK} C} \right) \right] \hat R 
    = \left[ 1-\sigma_E \left( \frac{B}{\theta_{RK} C} \right) \right] \left( \frac{\gamma\theta_{RK}}{D} \right) \hat s, \label{Soln:F} 
\end{align}
where we present $\hat F$ recursively as a function of $\hat R$ and in a closed-form solution. The sign of $\hat F$ is generally ambiguous since the bracketed term has ambiguous sign. That is, an increase in the offset price can \textit{either} increase or decrease fossil use depending on parameter values. To provide clarity on the sign of $\hat{F}$, Proposition \ref{Prop:FhatSign} identifies an if and only if condition where $\hat{s}>0$ implies $\hat{F}>0$. Generally speaking, the condition arises when $\sigma_X>>\sigma_E$ (noting $\sigma_X$ is in $C$); that is, the final goods sector is much more flexible regarding its inputs than the energy sector. We show such cases numerically in Section \ref{Sec:Numeric} where $\hat{F}>0$ and these cases relate to the issue of backfire, where increasing the offset price increases total emissions. Conversely, if $\sigma_E$ is large enough, then the offset induces sector $E$ to aggressively substitute towards $R$ and away from $F$, pushing the economy towards a reduction in total emissions.

\begin{proposition}[Sign of $\hat{F}$]
For $\hat{s}>0$, if $\sigma_R\neq0$ and $\sigma_X\neq0$ (implying $C\neq0)$, then there exist values of $\sigma_E>0$ such that $\hat{F}>0$. Furthermore, the range of $\sigma_E$ values leading to $\hat{F}>0$ is given by $\frac{\theta_{RK}C}{B}>\sigma_E\geq0$ such that $\hat{F}<0$ if and only if $\sigma_E>\frac{\theta_{RK}C}{B}>0.$ \label{Prop:FhatSign}
\end{proposition}
\begin{proof}
    Follows from Equation~\ref{Soln:F}.
\end{proof}

Continuing, the change in total energy production is given by:
\begin{align}
    \hat E = \left[ 1-\sigma_E \left( \frac{B}{\theta_{RK} C} \right) \theta_{EF} \right] \hat R, \label{Soln:E}
\end{align}
and we omit explicitly substituting for $\hat R$ going forward. Since $\hat F$ is generally ambiguous, then $\hat E$ is generally ambiguous as the production function for energy is given by $E=E(F,R)$. The comparative statics with respect to the elasticities for $\hat E$ are scaled versions those for $\hat F$, and thus generally ambiguous as well.

The changes in the final goods production sector are found to be:
\begin{align}
    \hat K_X =& \left( \frac{-1}{\alpha_X\theta_{RK}} \right) \left[ \alpha_R + \alpha_F\theta_{RK}\left( 1-\sigma_E \left( \frac{B}{\theta_{RK} C} \right) \right) \right] \hat R \label{Soln:KX}, \\
    \hat X = & \left[ \theta_{XE} - \left( \frac{\alpha_R+\alpha_F\theta_{RK}}{\alpha_X\theta_{RK}} \right) \theta_{XK} \right] \hat R = -\left(\frac{sR}{I} \right) \hat R \label{Soln:X}.
\end{align}
For $\hat K_X$, if the bracketed term is positive, then an increase in the offset price (and subsequently $\hat R$) leads to a decrease in capital use for final goods production due to the resource constraint in equation \ref{Eq:Resource} and additional capital input use in sector $R$ (as well as sector $F$, if $\hat F>0$). 

As shown in Section \ref{Sec:Welfare}, $\hat X$ is the marginal cost of the offset policy and thus the sign of $\hat X$ is important to overall welfare. At first, it appears the sign of $\hat X$ is ambiguous given the bracketed term in \ref{Soln:X}, but it is not. Rather, bracketed terms simplifies to $-sR/I$, where $sR$ is the total offset payment and $I$ is total income (equal to GDP). Therefore, $\hat{X}$ is proportional to the change in $R$ but scaled by the share of offset payments in the economy.
\begin{proposition}[Sign of $\hat{X}$]
If $\hat{s}>0$, then $\hat{X}<0$. \label{Prop:XhatSign}
\end{proposition}
\begin{proof}
    Follows from $\hat{R}>0$ and Equation~\ref{Soln:X}.
\end{proof}

\noindent Proposition \ref{Prop:XhatSign} means there is always a cost of the offset policy (i.e., no free lunch). This result holds even with as-if lump-sum taxes that fund the offset payments. In contrast, the marginal benefit of an offset payment can be negative (as discussed next).

Finally, the change in total emissions is $\hat Z = \rho_F \hat F + \rho_X \hat K_X$ and thus can be recursively written:
\begin{align}
    \hat Z = \rho_F \left[ 1-\sigma_E \left( \frac{B}{\theta_{RK} C} \right) \right] \hat R + \rho_X \left( \frac{-1}{\alpha_X\theta_{RK}} \right) \left[ \alpha_R + \alpha_F\theta_{RK}\left( 1-\sigma_E \left( \frac{B}{\theta_{RK} C} \right) \right) \right] \hat R. \label{Soln:Z}
\end{align}
The components of $\hat Z$ have ambiguous signs and thus $\hat Z$ does too since the emission shares $(\rho_F,\rho_R)$ are independent of other parameter values. Therefore, the marginal benefit of the offset policy has an ambiguous sign.
In the next section, we numerically demonstrate cases where $\hat{Z}$ is positive meaning the marginal benefit is negative, and in such cases welfare is unambiguously negative since $\hat{X}<0$. Finally, we note that the sign of $\hat{Z}$ is ambiguous even in the case where $\xi_X=0$ (that is, the case where there are no emissions from Sector~$X$); this is due to the ambiguous sign of $\hat{F}$, as demonstrated in Proposition~\ref{Prop:FhatSign}.

In this model, backfire $(\hat{Z}>0)$ can occur via two different qualitative paths due to the two channels of possible emission changes $(\hat{F},\hat{K}_X)$. In contrast, partial equilibrium models typically have only one channel for backfire ($\hat{F}>0$). Proposition~\ref{Prop:Backfire} catalogs the two paths. Condition (i) is similar to partial equilibrium backfire as it does not rely on the general equilibrium effects in Sector $X$. In Section~\ref{Sec:Numeric}, Table~\ref{Tab:Elasticity} finds emissions backfire under condition (i) using parameter values from the U.S. economy. Then, condition (ii) arises under a narrow set of parameter values, since backfire in this parameterization is completely driven by general equilibrium backfire  $(\hat{K}_X>0)$. Proposition~\ref{Prop:Backfire} provides a necessary condition for each backfire condition. From the equation $\hat Z = \rho_F \hat F + \rho_X \hat K_X$, it seems like there might be a third condition leading to backfire with $\hat{F}>0$ and $\hat{K}_X>0$. However, $\hat{F}>0$ implies $\hat{K}_X<0$ as documented in Proposition~\ref{Prop:Backfire}. That is, partial equilibrium backfire cannot be supported by general equilibrium backfire because increasing $F$ (and $R$) leaves fewer real resources for sector $X$. 

\begin{proposition}[Backfire] \label{Prop:Backfire}
If $\hat{Z}>0$, then one of the following conditions must hold: (i) $\hat{F}>0$ and $\hat{K}_X\leq0$; or, (ii) $\hat{F}\leq0$ and $\hat{K}_X>0$. Define $\left[ 1-\sigma_E \left( \frac{B}{\theta_{RK} C} \right) \right]\equiv\Theta $, and then a necessary condition for each backfire condition is as follows:

(i) $\Theta>0$

(ii) $\Theta<0$ and $\alpha_F\theta_{RK}|\Theta|>\alpha_R$.

\noindent Also, $\hat{F}>0$ and $\hat{K}_X>0$ cannot hold as a path for backfire.
\end{proposition}
\begin{proof}
    Conditions (i) and (ii) follows from $\hat{R}>0$ and equation~\ref{Eq:Zhat} as well as equations~\ref{Soln:F} and~\ref{Soln:KX}. Recall the resource constraint $0=\alpha_X \hat K_X + \alpha_R \hat K_R + \alpha_F \hat K_F$ noting that $\hat{R}>0$ implies $\hat{K}_R>0$ and $\hat{F}>0$ implies $\hat{K}_F>0$, and therefore $\hat{K}_X<0$ when $\hat{F}>0$.
\end{proof}

The final condition of Proposition~\ref{Prop:Backfire} says that it cannot be the case that both $\hat{F}>0$ and $\hat{K}_X>0$, implying that the two sectors have a dampening effect on the backfire of the other. Intuitively, this dampening occurs because there are finite emissions-producing inputs in the economy, and increasing those inputs in one sector necessarily reduces them in another.  


\section{Numerical Analysis} \label{Sec:Numeric}

This section assigns parameter values for the price-quantity model and welfare formulas. The purpose of conducting the numerical analysis is to generate intuition about the model as well as highlight key features and results and how they change as the underlying economic parameters vary.

To parameterize the model and by way of an illustrative example, we start by setting GDP equal to U.S. nominal GDP of $\$29.3$ trillion in 2024 \citep{bea2026gdp}. Next, we set all initial prices in the economy to one and allocate capital across sectors. Again, by way of example, we set the capital shares as $K_X=0.92$, $K_F=0.05$, and $K_R=0.01$ such that $\bar K=0.99 $ (implying $\bar Q = 0.01$ and total resources sum to one). Thus, sector $X$ uses most of the primary input while the renewable energy sector $R$ uses the least. These shares are not identical to the real input shares for the U.S. economy in 2024, but they are qualitatively accurate and thus means the result here are illustrative of real economic outcomes from a hypothetical carbon offset program.

Next, the initial \textit{ad valorem} tax on capital (i.e., the primary input) is set at $t=0.005$. The tax on capital must be small since all revenue goes directly to the offset payment, sector $R$ is small relative to the size of the economy, and there are no other government expenditures in the model. Solving for the initial offset price finds $s=0.197$ meaning the subsidized price of clean energy is $P_R+s=1.197$, so nearly 20 percent higher than the price of good $X$ as the numeraire. All non-pollution share parameters can be recovered from these level values. 

Continuing, we set $\rho_F=0.299$ and $\rho_X=0.701$ using data from the \citet{USEIA2025co2} on emissions on U.S. $\textrm{CO}_2$ emissions by sector. Total U.S. $\textrm{CO}_2$ in 2024 were 4,772 million metric tons (Mt) where 1,427 Mt were emitted by the electricity sector. Since we model carbon offset focused on the electricity sector, then $F$ is fossil-fuel electricity and $\rho_F=  \frac{1427}{4772}=0.299$, and $\rho_X=1-\rho_F$. For $\mu$, we use a social cost of carbon of \$200 per ton, based on the U.S. EPA's previous estimate of \$190 per ton in 2020 dollars \citep{USEPA2022}. Unless otherwise specified, the numerical simulations increase the offset price by 10 percent $(\hat s=10)$ and set $\phi=1$ (i.e., all initial offsets are additional by traditional measures).\footnote{We choose $\phi=1$ for two reasons. First, limited data exists regarding the true additionality of offsets in the US economy. Second, setting $\phi=1$ is a conservative, easy-to-interpret assumption. We show that even under this assumption (that all offsets are additional by traditional measures), conventional carbon accounting still frequently overstates aggregate emissions changes.}

Table \ref{Tab:Elasticity} reports changes in the quantity outcomes as well as the additionality ratio $\Delta$ (from Section \ref{Sec:Additionality}) and the welfare outcomes (from Section \ref{Sec:Welfare}), including calculations of $\mu^*$, as the elasticities vary.\footnote{The specific elasticity values our numerical exercises are for illustrative purposes only. However, these values are qualitatively similar to elasticity values found in \cite{goulder2017confronting}'s Table 4.3 that reports central parameter values for their general equilibrium study. We differ by providing several values for each parameter to show the range of outcomes in our model.} 
Each row of the table reports results from a different simulation that varies the elasticity values. Several observations arise from the results. First, $\hat R>0$ in all rows as found in Equation \ref{Soln:R}. Second, the sign of $\hat F$ can change. In particular, $\hat{F}>0$ for $\sigma_X > >\sigma_E$ leading to backfire and an increase in emissions (see rows 13-18). For these rows, note that $\Delta < - 1$. Third, $\hat X <0$ always, and thus numerically validates Proposition~\ref{Prop:XhatSign}. Fourth, comparing Panels A and B shows that a higher $\sigma_E$ yields higher $\Delta$ values, better welfare outcomes, and higher MVPF, although the change in total welfare is negative in all cases where the MVPF is below one. This mirrors findings by \cite{hahn2024welfare} that subsidies directly displacing dirty energy production tend to have a higher MVPF---in our setting, higher $\sigma_E$ implies a higher degree of substitution between $R$ and $F$, meaning increases in $R$ will more effectively displace $F$. Fifth, $\mu^*$ is near \$200 in the initial rows of the table indicating that the baseline $\mu=\$200$ is close to yielding non-negative total welfare. Higher $\mu$ values would increase the welfare gains from $Z$ in rows where $\hat Z<0$. Sixth, $\mu^*$ does not vary with $\sigma_R$.\footnote{This result is not apparent from equation \ref{Eq:MuStar}, but arises after substituting the closed-form expression from Section \ref{Sec:Analytical}; specifically, $\sigma_R$ only appears directly in $\hat{R}$, but $\hat{R}$ cancels out in the ratio $\hat{X}/\hat{Z}$ using equations \ref{Soln:X} and \ref{Soln:Z}.} The intuition for this result is that varying $\sigma_R$ only impacts the partial equilibrium supply of $R$ but varying $\sigma_X$ and $\sigma_E$ affects the general equilibrium the supply and demand for $R$ as well as $X$ and $Z$. 
Seventh, the MVPF is always less than one since the change in welfare is always negative. Finally, the MVPF is monotonic with the welfare value since MVPF is a benefit-cost ratio of the offset policy. 
Additional figures showing how emissions, additionality, and welfare vary with $\sigma_E$ and $\sigma_X$ are in Appendix~\ref{App:Figures}.

\begin{table}[t!]
    \centering
    \caption{Numerical Results Varying the Production Elasticities} \vspace{1mm}
    \resizebox{\columnwidth}{!}{%
    \begin{tabular}{c|cc|ccccc|c|rrr|c|c}
    \hline \hline
     & \multicolumn{2}{c|}{Elasticities} & \multicolumn{5}{c|}{Quantities (\% Change)} & Add. & \multicolumn{3}{|c|}{Welfare (\$billion)} & (\$) & (Index) \\    
    Row & $\sigma_R$ & $\sigma_X$ & $\hat R$ & $\hat F$ & $\hat E$ & $\hat X$ & $\hat Z$ & $\Delta$ & $X \;\,$ & $Z \;\,$ & Total & $\mu^*$ & MVPF \\
    \hline
    \multicolumn{14}{c}{\textbf{Panel A: $\sigma_E=3.00$ (high value)}} \\
    \hline
     1 & 0.25 & 0.25 & 0.686 & -0.310 & 0.022 & -0.003 & -0.097 & -0.06 & -99.5 & 92.2 & -7.3 & 215.9 & 0.995 \\
     2 & 0.75 & 0.25 & 1.537 & -0.695 & 0.050 & -0.008 & -0.216 & -0.06 & -223.0 & 206.5 & -16.5 & 215.9 & 0.990 \\
     3 & 1.50 & 0.25 & 2.228 & -1.007 & 0.072 & -0.011 & -0.314 & -0.06 & -323.2 & 299.3 & -23.9 & 215.9 & 0.987 \\
     4 & 0.25 & 0.75 & 0.694 & -0.244 & 0.069 & -0.003 & -0.080 & -0.23 & -100.7 & 75.9 & -24.8 & 265.3 & 0.984\\
     5 & 0.75 & 0.75 & 1.579 & -0.555 & 0.157 & -0.008 & -0.181 & -0.23 & -229.0 & 172.6 & -56.4 & 265.3 & 0.966\\
     6 & 1.50 & 0.75 & 2.316 & -0.814 & 0.230 & -0.011 & -0.265 & -0.23 & -335.9 & 253.2 & -82.7 & 265.3 & 0.954 \\
     7 & 0.25 & 1.50 & 0.705 & -0.158 & 0.130 & -0.003 & -0.057 & -0.46 & -102.2 & 54.8 & -47.5 & 373.3 & 0.969 \\
     8 & 0.75 & 1.50 & 1.635 & -0.367 & 0.301 & -0.008 & -0.133 & -0.46 & -237.1 & 127.0 & -110.0 & 373.3 & 0.935 \\
     9 & 1.50 & 1.50 & 2.438 & -0.548 & 0.448 & -0.012 & -0.199 & -0.46 & -353.7 & 189.5 & -164.2 & 373.3 & 0.909 \\
    \hline
    \multicolumn{14}{c}{\textbf{Panel B: Low $\sigma_E=0.25$ (low value)}} \\
    \hline
    10 & 0.25 & 0.25 & 0.302 & -0.009 & 0.094 & -0.001 & -0.009 & -0.79 & -43.7 & 8.9 & -34.9 & 984.3 & 0.977\\
    11 & 0.75 & 0.25 & 0.399 & -0.012 & 0.125 & -0.002 & -0.012 & -0.79 & -57.8 & 11.7 & -46.1 & 984.3 & 0.969 \\
    12 & 1.50 & 0.25 & 0.433 & -0.013 & 0.136 & -0.002 & -0.013 & -0.79 & -62.9 & 12.8 & -50.1 & 984.3 & 0.967 \\
    13 & 0.25 & 0.75 & 0.400 & 0.147 & 0.232 & -0.002 & 0.029 & -1.49 & -58.0 & -28.0 & -86.0 & - & 0.943 \\
    14 & 0.75 & 0.75 & 0.591 & 0.218 & 0.342 & -0.003 & 0.043 & -1.49 & -85.7 & -41.4 & -127.0 & - & 0.917 \\
    15 & 1.50 & 0.75 & 0.671 & 0.247 & 0.388 & -0.003 & 0.049 & -1.49 & -97.3 & -47.0 & -144.2 & - & 0.907 \\
    16 & 0.25 & 1.50 & 0.494 & 0.297 & 0.362 & -0.002 & 0.066 & -1.90 & -71.6 & -63.1 & -134.7 & - & 0.911 \\
    17 & 0.75 & 1.50 & 0.820 & 0.493 & 0.602 & -0.004 & 0.110 & -1.90 & -119.0 & -104.9 & -223.9 & - & 0.857 \\
    18 & 1.50 & 1.50 & 0.983 & 0.591 & 0.722 & -0.005 & 0.132 & -1.90 & -142.6 & -125.7 & -268.3 & - & 0.832 \\
    \hline \hline
    \multicolumn{14}{p{17cm}}{\small{Notes: These simulations normalize all initial prices equal to one with $\hat s =10$ and $\phi=1$. The MED is set at $\mu=200$. Also, the input resource share parameters are fixed at $\alpha_R=0.020$, $\alpha_F=0.051$, and $\alpha_X=0.929$.}}
    \end{tabular}}
    \label{Tab:Elasticity}
\end{table}

Some of the key results from Table~\ref{Tab:Elasticity} are also illustrated in Figure~\ref{fig:sigmaRX_emissions}, which shows how emissions vary with elasticities values, and Figure~\ref{fig:sigmaRX_welfare}, which shows how welfare varies with elasticities values. Figure~\ref{fig:sigmaRX_emissions} demonstrates that emissions increase for a wide range of $\sigma_R$ and $\sigma_X$ values when $\sigma_E$ is low, but decrease if $\sigma_E$ is high. In contrast, Figure~\ref{fig:sigmaRX_welfare} shows that welfare increases or decreases depending on both $\sigma_R$ and $\sigma_X$. If $\sigma_R$ or $\sigma_X$ are close to zero, then the change in welfare is negative, but small. However, if both $\sigma_R$ and $\sigma_X$ are large, then the change in welfare is substantially negative. 

\begin{figure}[t!]
     \centering
     \begin{subfigure}[b]{0.49\textwidth}
         \centering
         \includegraphics[width=\textwidth]{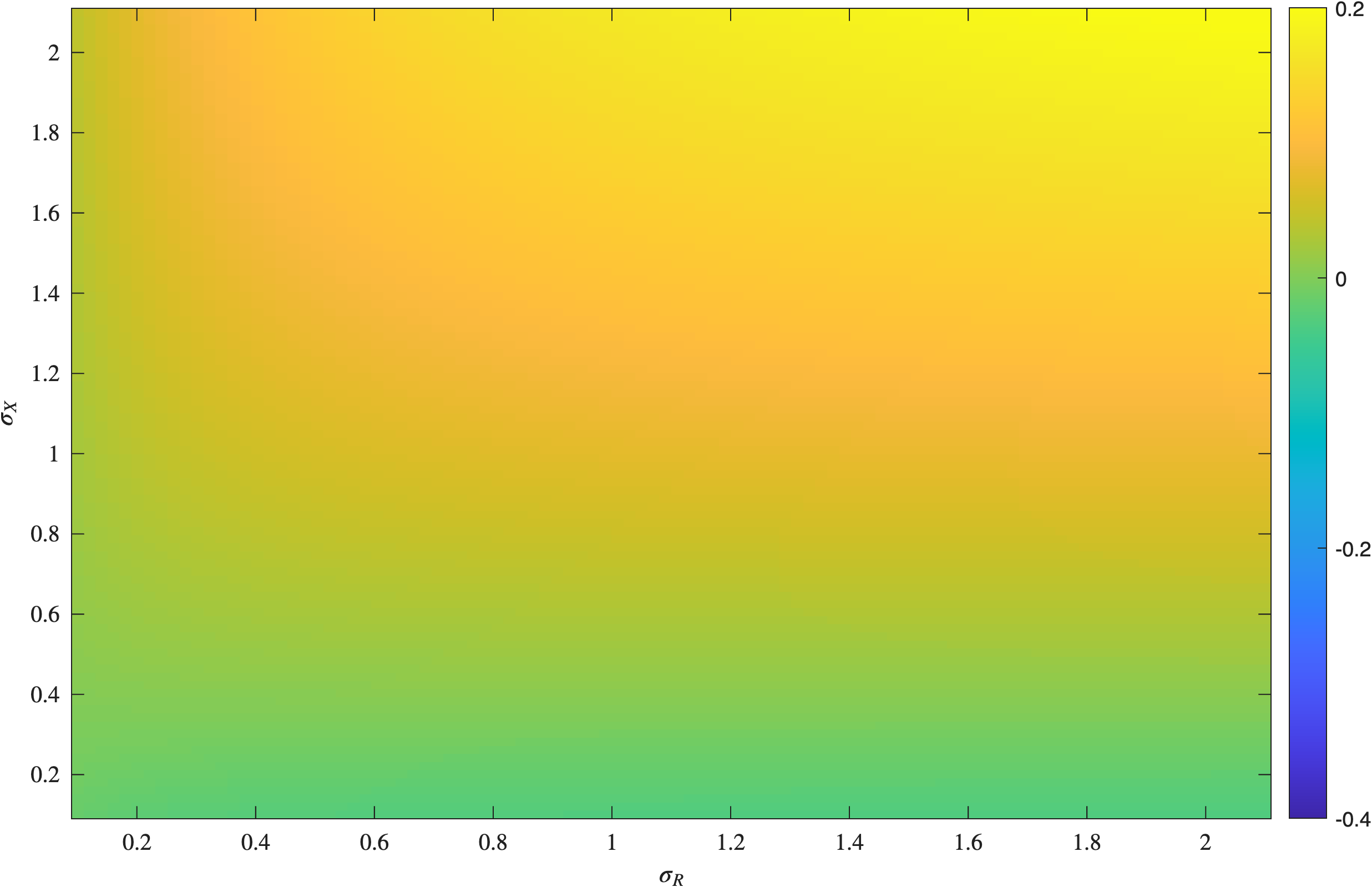}
         \caption{Low $\sigma_E=0.25$.}
         \label{fig:sigmaRX_emissions_lowE}
     \end{subfigure}
     \hfill
     \begin{subfigure}[b]{0.49\textwidth}
         \centering
         \includegraphics[width=\textwidth]{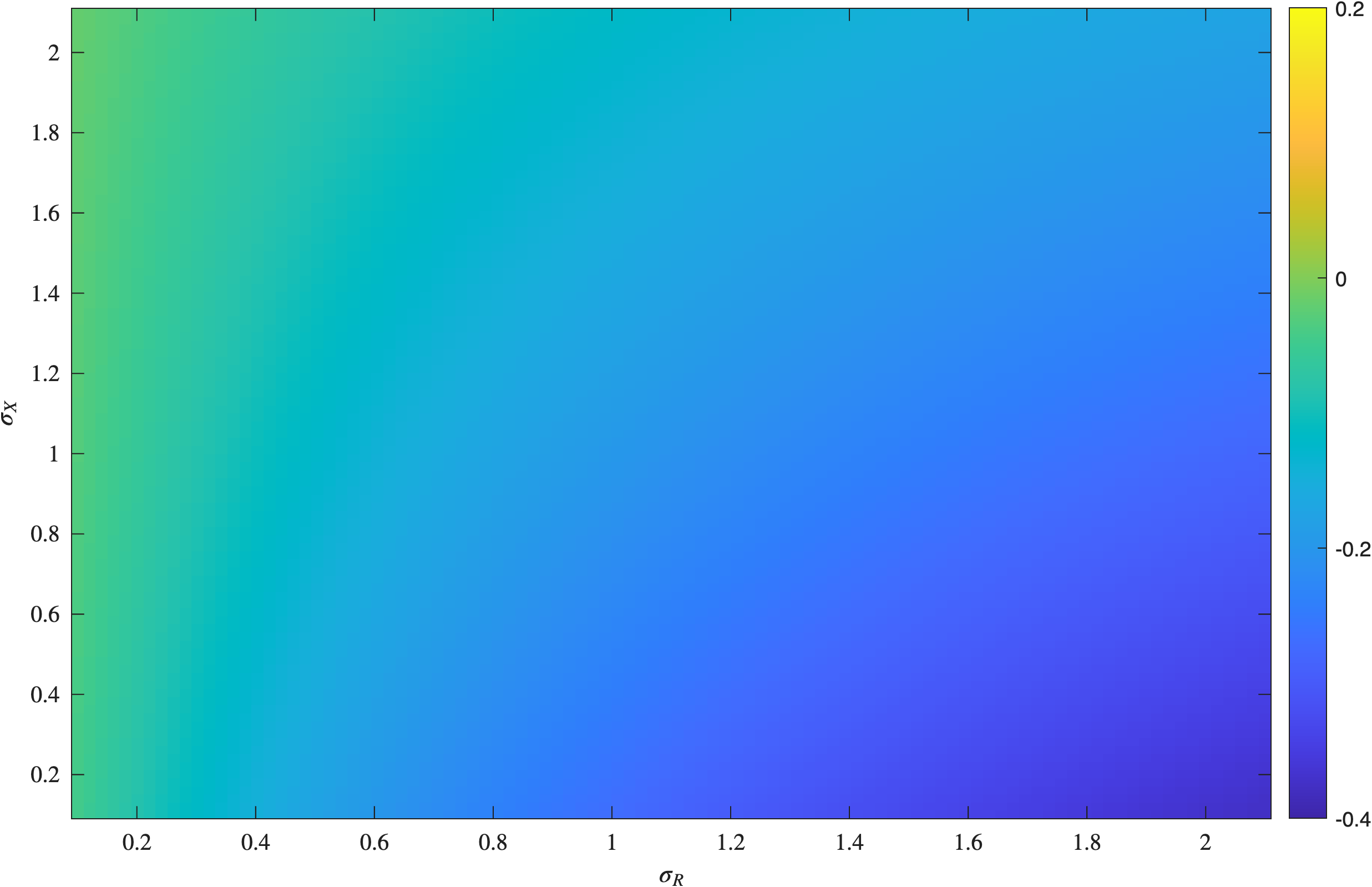}
         \caption{High $\sigma_E=3$.}
         \label{fig:sigmaRX_emissions_highE}
     \end{subfigure}
     \caption{Change in emissions as elasticities vary, $\hat{Z}$.}
     \label{fig:sigmaRX_emissions}
\end{figure}

\begin{figure}[t!]
     \centering
     \begin{subfigure}[b]{0.49\textwidth}
         \centering
         \includegraphics[width=\textwidth]{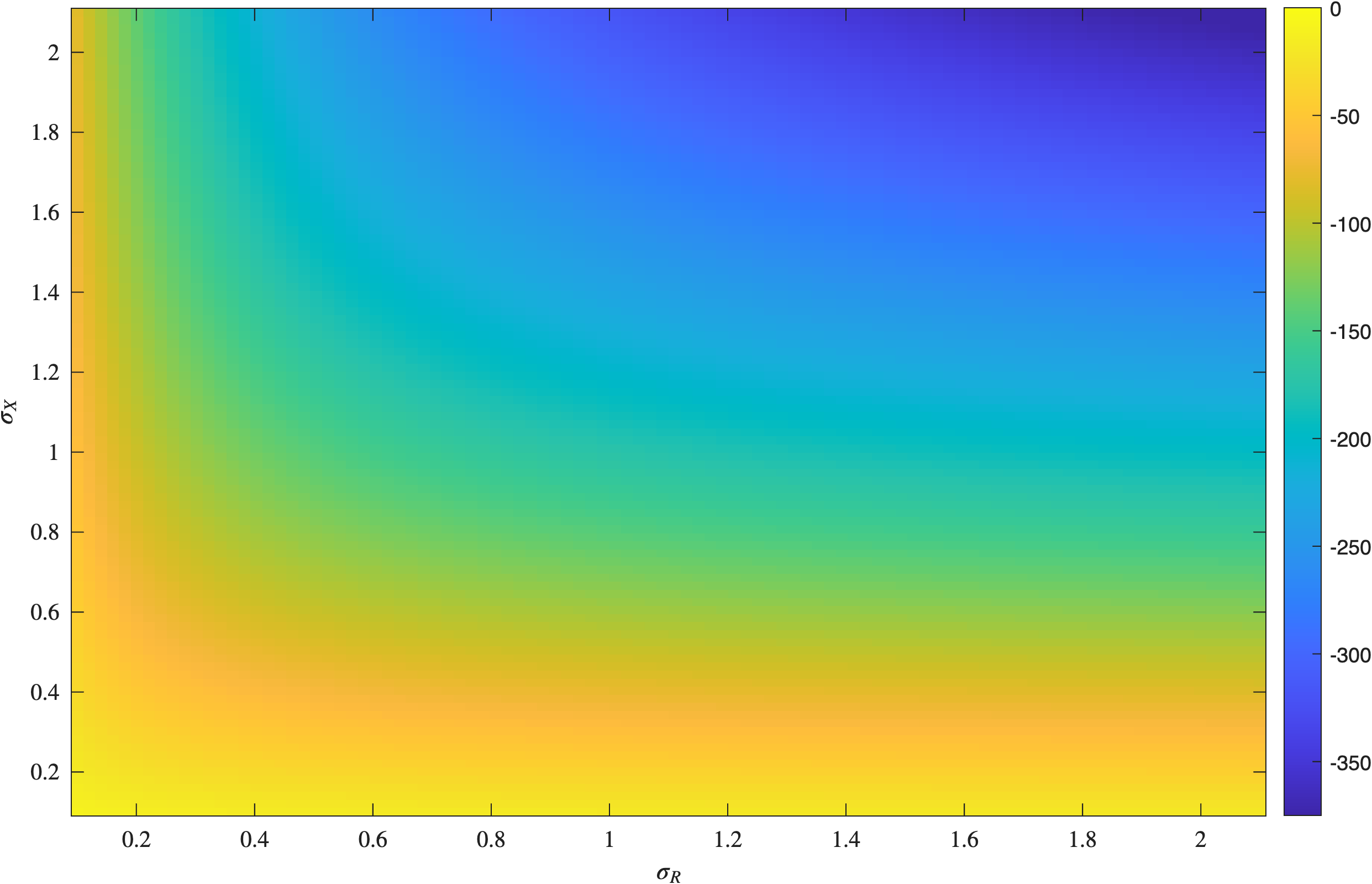}
         \caption{Low $\sigma_E=0.25$.}
         \label{fig:sigmaRX_welfare_lowE}
     \end{subfigure}
     \hfill
     \begin{subfigure}[b]{0.49\textwidth}
         \centering
         \includegraphics[width=\textwidth]{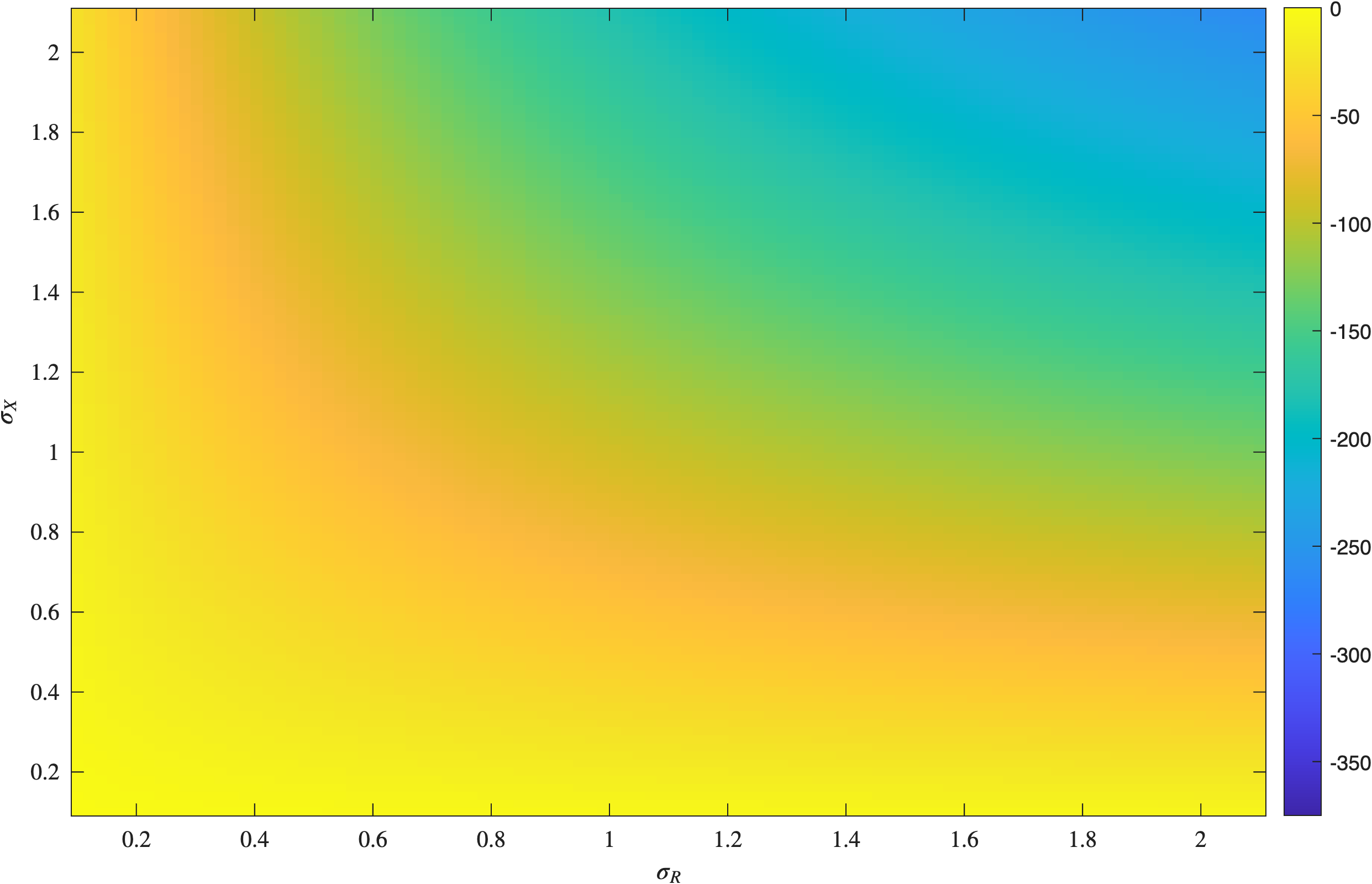}
         \caption{High $\sigma_E=3$.}
         \label{fig:sigmaRX_welfare_highE}
     \end{subfigure}
     \caption{Change in welfare as elasticities vary (billions of dollars).}
     \label{fig:sigmaRX_welfare}
\end{figure}

Next, Table \ref{Tab:Emission} reports results from changing the initial emission intensities by varying the underlying emission quantities across sectors $(Z_F,Z_X)$ while maintaining $4772=Z_F+Z_X$ in all rows. Changing the emission quantities then changes the emission intensities $(\xi_F,\xi_X)$ and the emission shares $(\rho_F,\rho_X)$ leading to different proportional emission changes $(\hat Z)$ even though the underlying economy is unchanged. That is, all rows in Table \ref{Tab:Emission} have the same change in carbon offsets $(\hat{R})$, fossil fuels $(\hat{F})$, total energy $(\hat{E})$, and final goods $(\hat X)$. (Recall that the change in emissions, $\hat{Z}$, is an auxiliary variable as defined in Section \ref{Sec:Model}.) Since $\hat X$ is fixed, then the change in welfare due to $X$ is the same in all rows. Rows 1-7 find a positive total welfare change despite $\Delta<0$. That is, the welfare change can be positive despite conventional carbon accounting falling short of aggregate emissions changes. These rows also have an MVPF greater than one. Row 8 solves for the emission intensities required such the costs of the offsets equal the benefits, and this finds $\mu^*=200$ by construction. Relatedly, since there is no change in total welfare, the MVPF equals one as the marginal willingness to pay (WTP) for offset program equals the marginal cost. Rows 9-14 have negative total welfare as found in Table \ref{Tab:Elasticity} and thus the MVPF is less than one. Interestingly, rows 12-14 report $\Delta>0$ meaning aggregate emission reductions are larger than the conventional calculation, but these rows have worse total welfare outcomes as a result of the policy change. Therefore, focusing just on additionality is a misleading metric for the welfare effect of a carbon offset policy. Portions of these results are also shown in Figure~\ref{fig:emissions_intensity_Delta}, which shows how $\Delta$ varies as the emissions intensity changes, with very high and positive values for $\Delta$ when Sector $F$'s initial emissions share is small, before rapidly decreasing and becoming negative as its share of emissions grows.

\begin{figure}[t!]
\centering
    \includegraphics[width=\textwidth]{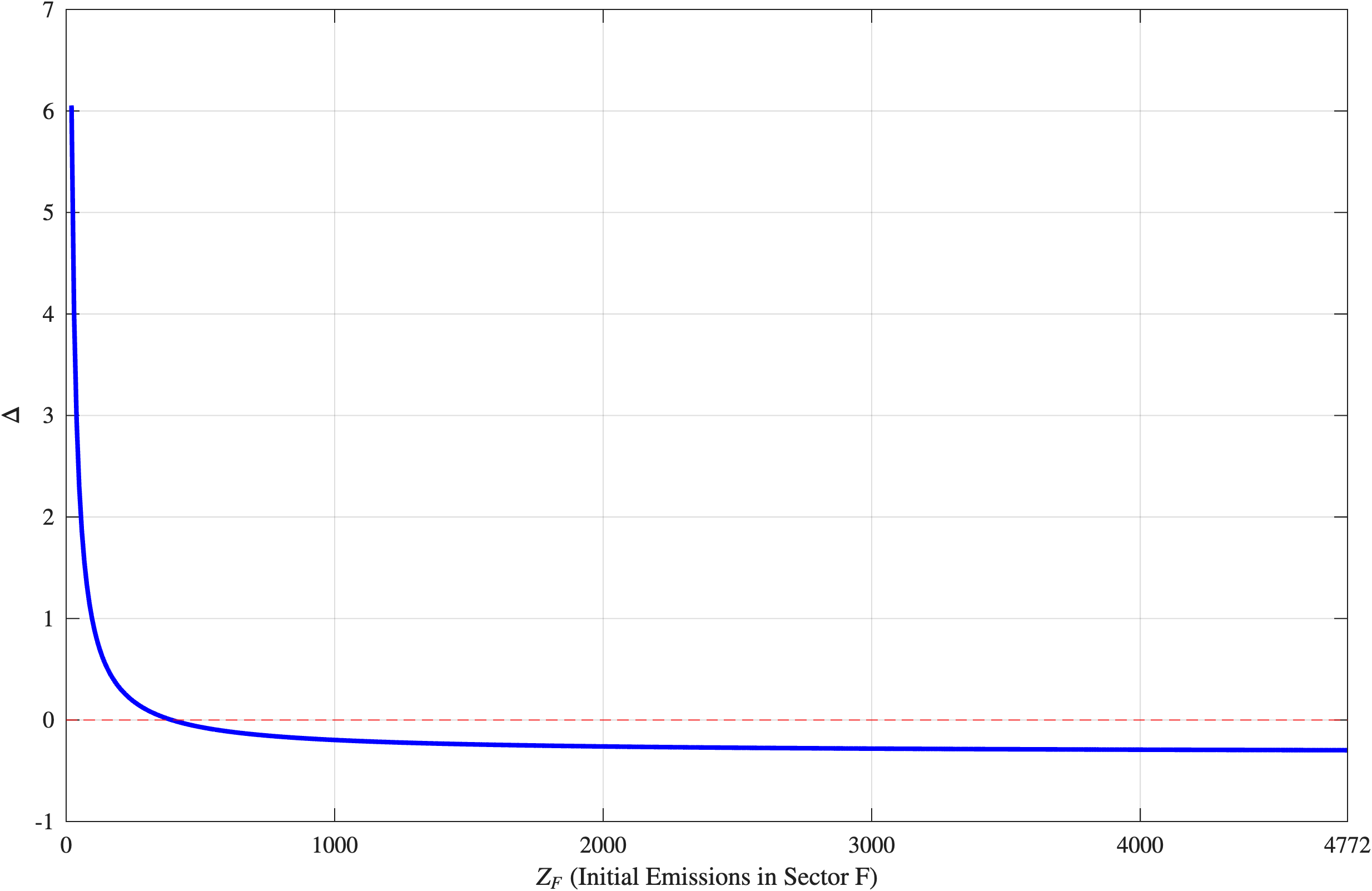}
    \caption{Change in $\Delta$ as emissions intensities vary.}
    \label{fig:emissions_intensity_Delta}
\end{figure}

The results in Table \ref{Tab:Emission} and Figure~\ref{fig:emissions_intensity_Delta} suggest that $\Delta$ can only be positive as $Z_F\rightarrow0$, but this is not the case. We repeat the numerical calculation in Row 1 (with $Z_F=4772$ and $Z_X=0$), but change the elasticity of substitution in the energy sector from $\sigma_E=3$ to $\sigma_E=10000$ (or nearly perfectly elastic). In this case, $\hat{F}<0$ and sufficiently negative to generate aggregate emission reductions larger than the conventional calculation using $\hat{R}>0$, and leading to $\Delta=0.015$. This result is driven by the ease that sector $E$ can switch from $F$ to $R$. This example numerically demonstrates the importance of $\sigma_E$ to the economic outcomes of an offset policy as well as further emphasizes the role of $\sigma_E$ in equations \ref{Soln:F} and \ref{Soln:E} that solve for $\hat{F}$ and $\hat{E}$, respectively. Similarly, from Tables \ref{Tab:Elasticity} and \ref{Tab:Emission}, one might infer that the backfire (i.e., $\hat{Z}>0$) can occur only if there are emissions from Sector $X$, but this is not the case either. To start, if $\xi_X=0$, then $\hat{Z}=\rho_F\hat{F}$ by Equation \ref{Eq:Zhat}.
Next, recall that Proposition \ref{Prop:FhatSign} finds that values of $\sigma_E$ always exist such that $\hat{F}>0$. Therefore, backfire can occur even if the rest of the economy $(X)$ emits zero pollution. 

\begin{table}[t!]
    \centering
    \caption{Numerical Results Varying the Emission Intensities} \vspace{1mm}
    \resizebox{\columnwidth}{!}{%
    \begin{tabular}{c|cc|ccccc|c|rrr|c|c}
    \hline \hline
     & \multicolumn{2}{c|}{Emissions} & \multicolumn{5}{c|}{Emission Parameters \& Change (\%)} & Add. & \multicolumn{3}{|c|}{Welfare (\$billion)} & (\$) & (Index) \\    
    Row & $Z_F$ & $Z_X$ & $\xi_F$ & $\xi_X$ & $\rho_F$ & $\rho_X$ & $\hat Z$ & $\Delta$ & $X \;\,$ & $Z \;\,$ & Total & $\mu^*$ & MVPF \\
    \hline
     1 & 4772 & 0 & 3257.3 & 0.0 & 1.000 & 0.000 & -0.244 & -0.298 & -100.7 & 233.0 & 132.3 & 86.5 & 1.085 \\
     2 & 4500 & 272 & 3071.7 & 10.1 & 0.943 & 0.057 & -0.231 & -0.296 & -100.7 & 220.2 & 119.5 & 91.5 & 1.077 \\
     3 & 4000 & 772 & 2730.4 & 28.6 & 0.838 & 0.162 & -0.206 & -0.292 & -100.7 & 196.7 & 96.0 & 102.4 & 1.062 \\
     4 & 3500 & 1272 & 2389.1 & 47.2 & 0.733 & 0.267 & -0.182 & -0.288 & -100.7 & 173.3 & 72.5 & 116.3 & 1.047 \\
     5 & 3000 & 1772 & 2047.8 & 65.7 & 0.629 & 0.371 & -0.157 & -0.282 & -100.7 & 149.8 & 49.1 & 134.5 & 1.032 \\
     6 & 2500 & 2272 & 1706.5 & 84.3 & 0.524 & 0.476 & -0.132 & -0.273 & -100.7 & 126.3 & 25.6 & 159.5 & 1.017 \\
     7 & 2000 & 2772 & 1365.2 & 102.8 & 0.419 & 0.581 & -0.108 & -0.261 & -100.7 & 102.8 & 2.1 & 195.9 & 1.001 \\
     8 & 1955 & 2817 & 1334.4 & 104.5 & 0.410 & 0.590 & -0.106 & -0.259 & -100.7 & 100.7 & 0.0 & 200.0 & 1.000 \\
     9 & 1500 & 3272 & 1023.9 & 121.4 & 0.314 & 0.686 & -0.083 & -0.239 & -100.7 & 79.3 & -21.4 & 253.8 & 0.986 \\
     10 & 1000 & 3772 & 682.6 & 139.9 & 0.210 & 0.790 & -0.059 & -0.197 & -100.7 & 55.9 & -44.8 & 360.5 & 0.971 \\
     11 & 500 & 4272 & 341.3 & 158.5 & 0.105 & 0.895 & -0.034 & -0.069 & -100.7 & 32.4 & -68.3 & 621.9 & 0.956 \\
     12 & 250 & 4522 & 170.6 & 167.8 & 0.052 & 0.948 & -0.022 & 0.186 & -100.7 & 20.7 & -80.1 & 975.3 & 0.948 \\
     13 & 100 & 4672 & 68.3 & 173.3 & 0.021 & 0.979 & -0.014 & 0.952 & -100.7 & 13.6 & -87.1 & 1480.2 & 0.944 \\
     14 & 10 & 4762 & 6.8 & 176.7 & 0.002 & 0.998 & -0.010 & 12.434 & -100.7 & 9.4 & -91.3 & 2146.9 & 0.941 \\
    \hline \hline
    \multicolumn{14}{p{17cm}}{\small{Notes: These simulations normalize all initial prices equal to one with $\hat{s}=10$ and $\phi=1$. The initial emission intensities vary by row. Total emissions are fixed such that $4772=Z=Z_F+Z_X$. The elasticity parameters are set at $\sigma_R=0.25$, $\sigma_X=0.75$, and $\sigma_E=3.00$. All other parameters are set to the values in Table \ref{Tab:Elasticity}. For all row, the quantity changes are $\hat{R}=0.694$, $\hat{F}=-0.244$, $\hat{E}=0.069$, and $\hat{X}=-0.003$.}}
    \end{tabular}}
    \label{Tab:Emission}
\end{table}

Finally, Table \ref{Tab:Subsidy} and Figure~\ref{fig:subsidy_welfare} show results from changing the initial capital tax and offset rates. In the prior tables, the initial \textit{ad valorem} capital tax is set at $t=0.005$ implying the per unit offset is $0.197$ (and repeated in row 5 here). Note that the initial capital tax is small in absolute terms because the capital tax base is large relative to the renewable sector. Table \ref{Tab:Subsidy} row 1 changes the initial tax to its smallest value, $t=0.001$, and finds a positive total welfare change---despite a negative additionality ratio. Also, $\mu^*=48.1$ means the per unit benefit of abatement would not need to be close to the social cost of carbon in our study $(\$200)$ to yield positive welfare gains. This result arises as the marginal cost of increasing the offset price $(\hat X)$ is small when the initial tax is small; that is, the height of the deadweight loss triangle is initially small. In contrast, row 10 has an initial tax ten-times larger than row 1, implying the initial deadweight loss is larger and leading to larger costs from increasing the offset price. In this case, the welfare change is negative and $\mu^*=607.3$, notwithstanding the proportional reduction in emissions being much larger in the $t=0.010$ case. Table \ref{Tab:Subsidy} further demonstrates the monotonic relationship between MVPF and welfare with $\textrm{MVPF}>1$ implying positive welfare changes and $\textrm{MVPF}<1$ implying negative welfare changes. When interpreting Table \ref{Tab:Subsidy}, one should take caution when comparing results across rows as the initial cost shares vary as the initial tax and offset price vary. In other words, each row represents a different economy.

\begin{figure}[t!]
\centering
    \includegraphics[width=\textwidth]{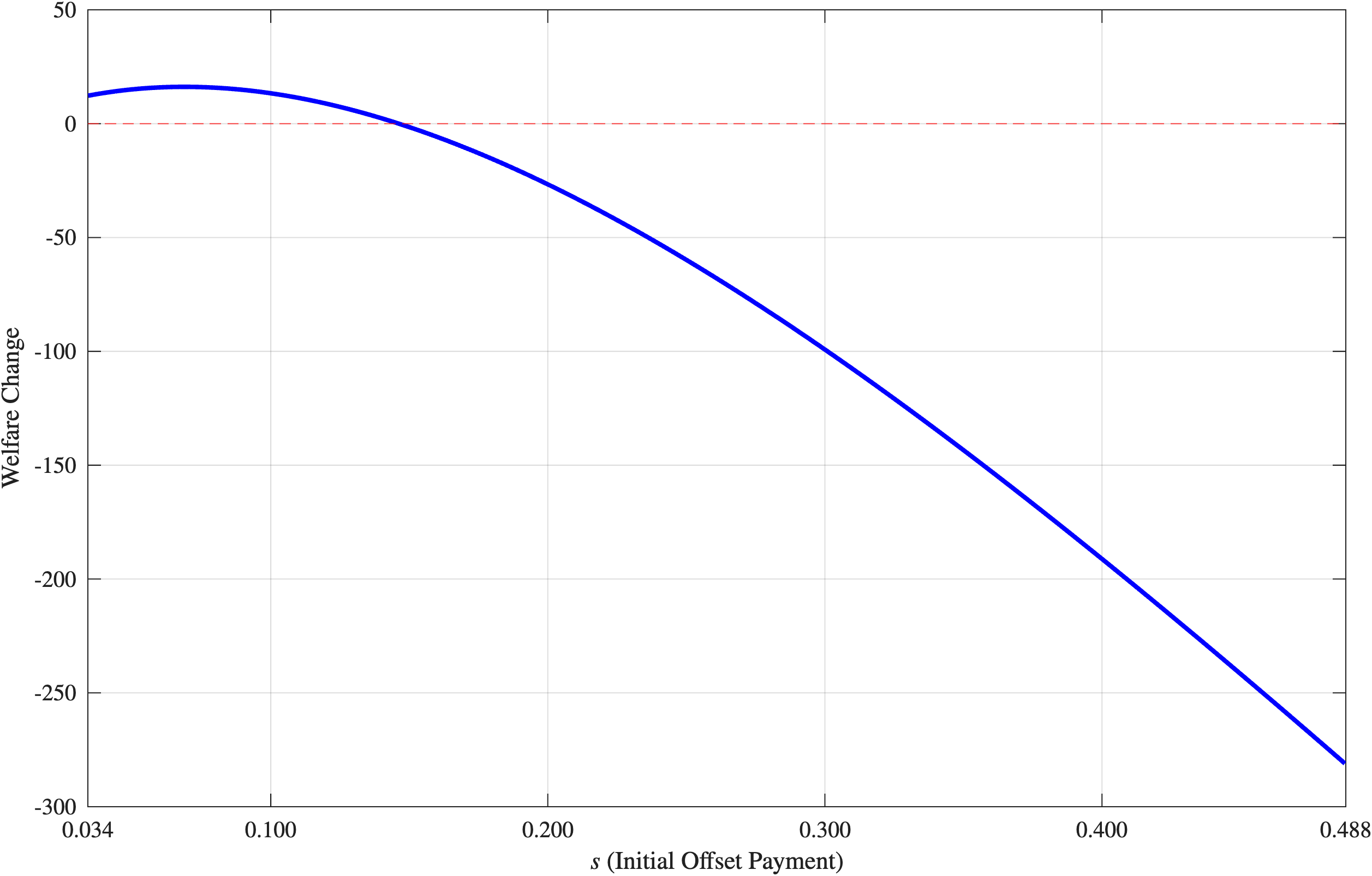}
    \caption{Change in welfare as initial offset payment level varies.}
    \label{fig:subsidy_welfare}
\end{figure}

\begin{table}[t!]
    \centering
    \caption{Numerical Results Varying the Initial Tax and Offset Payment Values} \vspace{1mm}
    \resizebox{\columnwidth}{!}{%
    \begin{tabular}{c|cc|ccccc|c|rrr|c|c}
    \hline \hline
     & \multicolumn{2}{c|}{Offset/Tax} & \multicolumn{5}{c|}{Quantities (\% Change)} & Add. & \multicolumn{3}{|c|}{Welfare (\$billion)} & (\$) & (Index) \\    
    Row & $s$ & $t$ & $\hat R$ & $\hat F$ & $\hat E$ & $\hat X$ & $\hat Z$ & $\Delta$ & $X \;\,$ & $Z \;\,$ & Total & $\mu^*$ & MVPF \\
    \hline
     1 & 0.034 & 0.001 & 0.135 & -0.053 & 0.016 & -0.0001 & -0.017 & -0.274 & -3.9 & 16.2 & 12.3 & 48.2 & 1.042 \\
     2 & 0.071 & 0.002 & 0.271 & -0.104 & 0.031 & -0.0005 & -0.033 & -0.265 & -15.7 & 31.9 & 16.2 & 98.7 & 1.027 \\
     3 & 0.110 & 0.003 & 0.410 & -0.153 & 0.044 & -0.0012 & -0.049 & -0.255 & -35.7 & 47.1 & 11.4 & 151.6 & 1.013 \\
     4 & 0.152 & 0.004 & 0.551 & -0.200 & 0.057 & -0.0022 & -0.065 & -0.245 & -64.0 & 61.8 & -2.2 & 207.0 & 0.998 \\
     5 & 0.197 & 0.005 & 0.694 & -0.244 & 0.069 & -0.0034 & -0.080 & -0.235 & -100.7 & 75.9 & -24.8 & 265.3 & 0.984 \\
     6 & 0.246 & 0.006 & 0.840 & -0.286 & 0.080 & -0.0050 & -0.094 & -0.224 & -146.1 & 89.5 & -56.7 & 326.7 & 0.970 \\
     7 & 0.299 & 0.007 & 0.987 & -0.325 & 0.089 & -0.0068 & -0.107 & -0.212 & -200.4 & 102.4 & -98.0 & 391.4 & 0.956 \\
     8 & 0.356 & 0.008 & 1.136 & -0.361 & 0.097 & -0.0090 & -0.120 & -0.199 & -263.7 & 114.7 & -149.0 & 459.9 & 0.942 \\
     9 & 0.419 & 0.009 & 1.288 & -0.395 & 0.104 & -0.0115 & -0.132 & -0.186 & -336.3 & 126.3 & -210.0 & 532.4 & 0.929 \\
     10 & 0.488 & 0.010 & 1.442 & -0.425 & 0.110 & -0.0143 & -0.144 & -0.171 & -418.3 & 137.3 & -281.0 & 609.4 & 0.915 \\
    \hline \hline
    \multicolumn{14}{p{17cm}}{\small{Notes: These simulations normalize all initial prices equal to one with $\hat{s}=10$ and $\phi=1$. The MED is set at $\mu=200$, and the elasticities are $\sigma_R=0.25$, $\sigma_X=0.75$, and $\sigma_E=3.00$. Also, the input resource share parameters are fixed at $\alpha_R=0.020$, $\alpha_F=0.051$, and $\alpha_X=0.929$. The initial cost shares vary by row.}}
    \end{tabular}}
    \label{Tab:Subsidy}
\end{table}

\subsection{Varying Initial Additionality Share} \label{Sec:VaryPhi}

All prior numerical simulations hold $\phi=1$, so this section analyzes how $\Delta$ changes when the initial share of additionality changes (recalling $\phi \equiv R^A/R$). Figure~\ref{fig:DeltaPhi} graphs $\Delta$ as a function of $\phi$ for four different elasticity cases, holding all other parameters constant, with share parameters the same as employed in Table~\ref{Tab:Elasticity}. We only graph values of $\phi \in [0.05,1]$ since $\Delta \rightarrow |\infty|$ as $\phi \rightarrow0$. For each case, we record the value (if it exists) where $\Delta=0$, that is, the value of $\phi$ for which conventional carbon accounting accurately measures aggregate emissions changes.

The four different elasticity cases represent three different qualitative outcomes. Case 1 is the Standard case with $\sigma_R=0.25$, $\sigma_X=0.75$, and $\sigma_E=3.00$. For this case, $\Delta=-0.25$ when $\phi=1$, see Table~\ref{Tab:Elasticity} row 4. Then, Case 1 intersects the horizontal axis at $\phi=0.765$ with $\Delta \approx 15$ at $\phi=0.05$. Recall that $\Delta>0$ implies conventional carbon accounting under-counts actual emission reductions.

Case 2 is similar to Case 1 in that $\Delta<0$ for $\phi=1$ and crosses the horizontal axis, but at a much lower $\phi$ value of 0.206. This Minimal change case induces only small quantity changes due to the small elasticity values $(\sigma_E=\sigma_X=\sigma_R=0.25)$. The Minimal change case corresponds to Table~\ref{Tab:Elasticity} row 10 and this row contains the smallest $\hat{Z}$ value reported in the table. Note that $\hat{Z}$ does not change when $\phi$ changes. Case 2 also shows that not all horizontal intercepts are identical for cases that cross the axis.

Case 3 never intersects the horizontal axis and $\Delta$ is always positive. We call this the Non-reactive case due to its small Sector $X$ elasticity value $(\sigma_X=0.10)$ and high substitutability of $F$ and $R$ ($\sigma_E = 5$). This is the case closest to satisfying the implicit assumptions under conventional carbon accounting: renewables can easily replace fossil fuels, and Sector $X$ does not substantially alter its inputs in response. Hence, conventional carbon accounting understates aggregate emissions change (i.e, true emission reductions) even when all initial offsets are additional by traditional measures $(\phi=1)$; this is because conventional accounting closely matches total emission and reductions from Sector~$X$ are small. Reducing $\phi$ increases $\Delta$ as in the previous cases. The energy sector’s large elasticity value $(\sigma_E=5)$ also contributes to this unusual scenario of persistently under-counting emissions reductions for all values of $\phi$.

Case 4 never intersects the horizontal axis as well, but we construct Case 4 such that $\Omega_A>0$ implying backfire. Also, $\Delta < -1$ always and lower $\phi$ leads to larger negative values. The Backfire case corresponds to Table~\ref{Tab:Elasticity} row 14 with $\hat{Z}=0.043$. In summary, Figure~\ref{fig:DeltaPhi} demonstrates how the initial share of offsets that are additional ($\phi$) impacts $\Delta$. It also shows how the underlying economic parameters are key to this relationship, as differing values of $\phi$ can lead to wildly different $\Delta$ values depending on the underlying economy.

\begin{figure}[t!]
\centering
    \includegraphics[width=\textwidth]{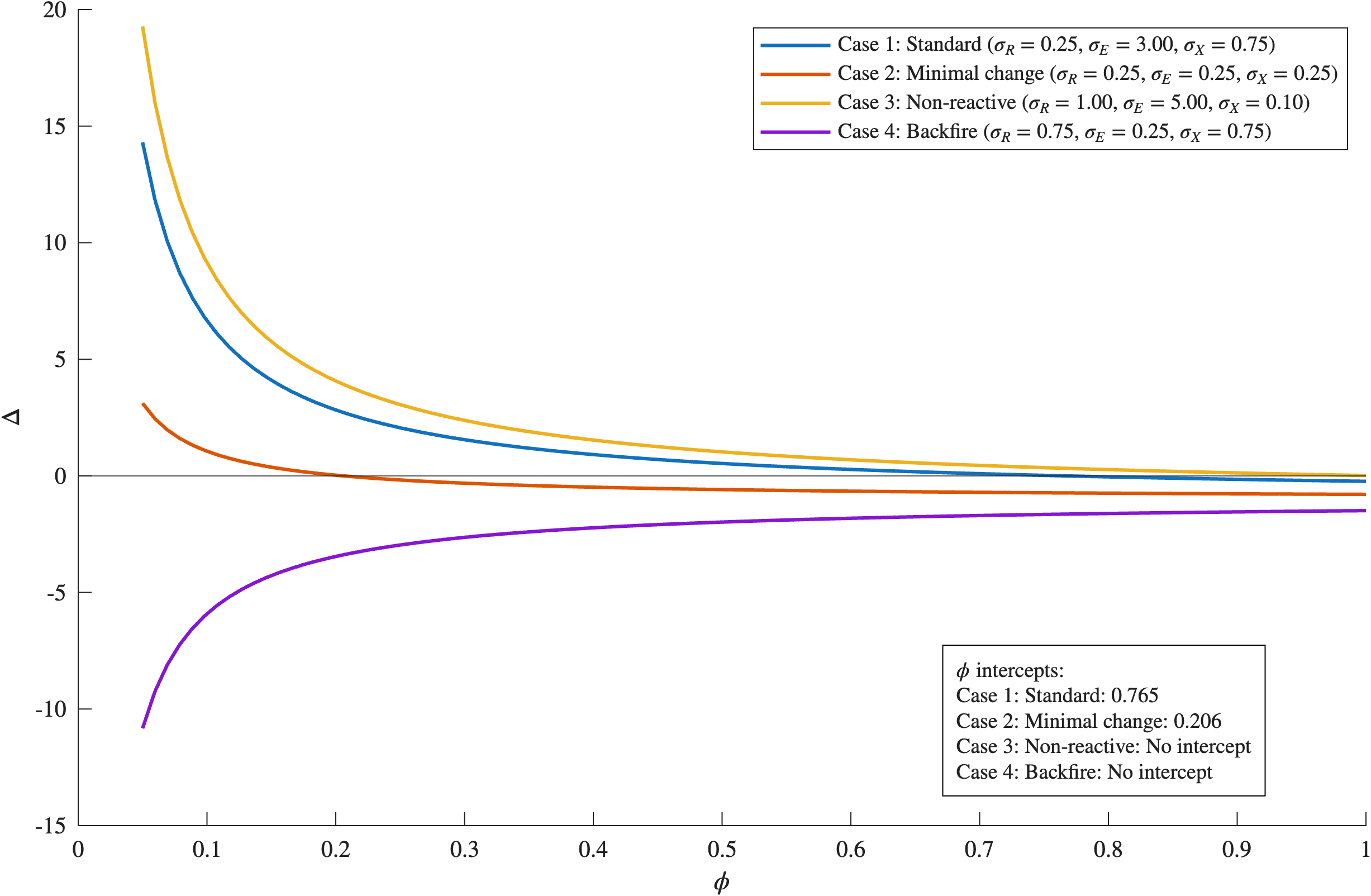}
    \caption{Change in additional as initial additionality share varies.}
    \label{fig:DeltaPhi}
\end{figure}

\subsection{Case Study: Carbon Capture and Storage (CCS)} \label{Sec:CCS}

The Energy Improvement and Extension Act of 2008 established the U.S. federal tax credit known as ``45Q''---in reference to its section in the federal tax code---to incentivize carbon capture and storage (CCS). Both the Inflation Reduction Act (IRA) of 2022 and the One Big Beautiful Bill (OBBB) of 2025 expanded the 45Q policy, suggesting that CCS is politically popular. In contrast, direct carbon pricing by the U.S. federal government has not occurred to date. At its core, 45Q is a subsidy for carbon-free, fossil-fuel-based electricity generation that operates much in the same way as a government-funded carbon offset purchase would operate.\footnote{Additionally, 45Q subsidizes enhanced oil recovery using CO$_2$, but we focus on 45Q's electricity market effects in this study. Also, we abstract from issues regarding the long-term stability of carbon sequestration.}
Thus, the 45Q tax credit fits well in our model to assess the additionality and emissions impacts of CCS policy.

The 2025 U.S. Annual Energy Outlook (AEO) provides projections of CCS outcomes in the electricity market under the IRA's 45Q policy, and thus we use the AEO projections to parameterize the model \citep{USEIA2025aeo}. In this case study, sector $R$ is all fossil-fuel fired electric generating units (EGUs) with CCS and $F$ is all fossil-fuel fired EGUs without carbon sequestration, so $E$ is total fossil electricity production and $X$ is all other sectors. In this setting, $\bar{Q}$ can be interpreted as the specific geological features required to make CCS feasible. Currently, CCS is not a prominent feature of the U.S. electric power sector, so we use AEO 2025's forecast for 2040 as the year of study. 

Under the 2025 AEO's reference scenario forecast, total 2040 $\textrm{CO}_2$ emissions from fossil-fuel EGUs is 490.5 million metric tons (MMmt) while CSS-enabled EGUs capture 38.0 MMmt. We model $R$ as emission neutral and thus 452.5 MMmt in  net emissions are attributed to $F$. For context, these carbon emissions arise from a base of 2,232 gigawatts (GW) of total electric power sector capacity generating 5,472 billion kWh.\footnote{For comparison, the 2027 forecast is 1,153.5 MMmt from $E$ (with zero CCS emission reductions) from 1,300 GW of capacity that generates 4,322 billion KWh. Thus, the 2025 AEO projects significant de-carbonization of the U.S. electric power sector without considering the net reductions from CCS.}
However, fossil-fuel based generation is only 1,284 billion kWh or 23.4\% of total electricity production. 
Total, economy-wide U.S. $\textrm{CO}_2$ emissions are projected to be 3456.3 MMmt and thus implying sector $X$ (i.e., rest of the economy) accounts for 3002.8 MMmt.

In terms of GDP and sectoral revenue, the 2040 GDP projection is \$38.97 trillion (in 2024 dollars). The 2040 average electricity price is 11.5 cents/kWh, so 1,284 billion kWh from fossil-fuel based generation implies \$149.7 billion in revenue (i.e., $p_EE=149.7$). For $R$ (i.e., CCS-enabled EGUs), the subsidy payment is \$85 per metric ton $\textrm{CO}_2$ sequestered \citep{CBO-Q45}, and thus leading to a total subsidy payment of \$3.2 billion to sequester 38.0 MMmt.\footnote{The Congressional Budget Office (CBO) report finds a similar per year cost estimate and states that the subsidy will be paid to a CCS-enabled EGU if ``Construction Beginning Prior to 1/1/2033''. The subsidy program is politically popular currently, so it is reasonable to assume that all of CCS-enabled EGUs in 2040 will be receiving the subsidy.} The AEO does not provide generation projections from CCS-enabled EGUs separately from total total fossil-fuel based generation, so we assume $p_EE$ is split between $R$ and $F$ by their share of (absolute) emissions. Under these assumptions, total revenue in $R$, denoted $(P_R+s)R$, is \$14.8 billion.\footnote{Specifically, the calculation is $149.7 \times (|\minus38|/490.5)+3.2=14.8$ (after independent rounding).} 
Then, as a residual, $P_FF=138.1$. Finally, we assume that one-third of the CCS unit's cost come from the fixed factor $(P_Q\bar{Q})$ to secure sequestration sites. The implied per unit subsidy is $s=0.279$ while the \textit{ad valorem} tax $(t)$ on capital is only a 0.008\% since sector $R$ is small.\footnote{As result, $\hat{X}$ in the numerical analysis is quite small and thus omitted from Table \ref{Tab:Sequestration}.}

Table \ref{Tab:Sequestration} reports numerical results when varying the production elasticities (in a similar manner as Table~\ref{Tab:Elasticity}). For this exercise, we assume $\phi=1$ which means that all initial CCS capacity is additional by traditional measures because CCS would not be economically feasible without a subsidy or a significant carbon tax \citep{CBO-CCS}. Also, we fix $\sigma_E=3$ at for Table \ref{Tab:Sequestration} to reflect the high degree of substitution between CCS-enabled and non-CCS EGUs.\footnote{We also note that $\sigma_R$ is likely to be relatively low since CCS requires specific geological features that are not easily substitutable with technology \citep{netl_fAQs}.} For this case study, the change in welfare is always small and positive and, relatedly, the MVPF is greater than one and $\mu^*<\mu=200$. This net increase in welfare is due to the small initial tax rate $(t)$ such that increasing the subsidy by 10 percent does not lead to large distortions in $X$. The low absolute value of the increase in welfare comes from the small size of sector $R$ such that increasing the subsidy has a small total effect on the economy. Regardless, the emission reductions in sector $F$ are counteracted by an increase in emissions from sector $X$ as shown by $\Delta<0$; that is, not all of the offsets are additional in aggregate. For rows 7-9, $\Delta=-0.50$ meaning half of the emission reductions in $R$ due to the 45Q policy are erased by a corresponding increase in emissions elsewhere. 

\begin{table}[t!]
    \centering
    \caption{Carbon Sequestration Case Study Results} \vspace{1mm}
    \resizebox{\columnwidth}{!}{%
    \begin{tabular}{c|cc|cccc|c|rrr|c|c}
    \hline \hline
     & \multicolumn{2}{c|}{Elasticities} & \multicolumn{4}{c|}{Quantities (\% Change)} & Add. & \multicolumn{3}{|c|}{Welfare (\$billion)} & (\$) & (Index) \\    
    Row & $\sigma_R$ & $\sigma_X$ & $\hat R$ & $\hat F$ & $\hat E$ & $\hat Z$ & $\Delta$ & $X \;\,$ & $Z \;\,$ & Total & $\mu^*$ & MVPF \\
    \hline
     1 & 0.25 & 0.25 & 0.955 & -0.073 & 0.007 & -0.010 & -0.08 & -3.1 & 6.7 & 3.6 & 92.4 & 1.102 \\
     2 & 0.75 & 0.25 & 2.299 & -0.176 & 0.016 & -0.023 & -0.08 & -7.4 & 16.1 & 8.7 & 92.4 & 1.218 \\
     3 & 1.50 & 0.25 & 3.548 & -0.272 & 0.024 & -0.036 & -0.08 & -11.5 & 24.8 & 13.4 & 92.4 & 1.305 \\
     4 & 0.25 & 0.75 & 0.957 & -0.059 & 0.020 & -0.008 & -0.25 & -3.1 & 5.4 & 2.4 & 113.4 & 1.067 \\
     5 & 0.75 & 0.75 & 2.309 & -0.143 & 0.047 & -0.019 & -0.25 & -7.5 & 13.1 & 5.7 & 113.4 & 1.143 \\
     6 & 1.50 & 0.75 & 3.570 & -0.221 & 0.073 & -0.029 & -0.25 & -11.5 & 20.3 & 8.8 & 113.4 & 1.201 \\
     7 & 0.25 & 1.50 & 0.959 & -0.039 & 0.038 & -0.005 & -0.50 & -3.1 & 3.7 & 0.6 & 169.5 & 1.016 \\
     8 & 0.75 & 1.50 & 2.322 & -0.094 & 0.093 & -0.013 & -0.50 & -7.5 & 8.9 & 1.4 & 169.5 & 1.034 \\
     9 & 1.50 & 1.50 & 3.603 & -0.146 & 0.144 & -0.020 & -0.50 & -11.6 & 13.7 & 2.1 & 169.5 & 1.048 \\
    \hline \hline
    \multicolumn{13}{p{15cm}}{\small{Notes: These simulations normalize all initial prices equal to one with $\hat{s}=10$ and $\phi=1$. The MED is set at $\mu=200$. Also, the input resource share parameters are fixed at $\alpha_R=0.0003$, $\alpha_F=0.0035$, and $\alpha_X=0.9962$. We set $\sigma_E=3$ for all rows.}}
    \end{tabular}}
    \label{Tab:Sequestration}
\end{table}


\section{Conclusion}\label{sec:conclusion}

This study examines carbon offset policies in general equilibrium using an analytical model to find closed-form solutions while focusing on additionality, emissions, and welfare effects. The general equilibrium setting allows for the comparison between the direct effect of carbon offsets on emissions (i.e., ``conventional carbon accounting'') and the change in economy-wide emissions (i.e., ``aggregate carbon accounting'') for a given offset policy. Also, we decompose and categorize offsets into four types and show how they respond to payments, including one type not previously identified in the literature. Specifically, we demonstrate how offsets that are non-additional by traditional measures still respond to offset price changes on the margin. When this group of offsets is large, conventional additionality can understate aggregate emissions changes, and it is this offset type that has not been previously identified. 

What are the lessons for policymakers? First, given an offset payment increase, conventional carbon accounting is not a good proxy measure for welfare change since carbon offsets can decrease welfare, particularly, when the initial offset price is high, when the initial emissions from the dirty sector are low, and when clean and dirty energy are relatively complementary (i.e., not good substitutes).\footnote{We note that many carbon offset prices are currently low. Offset industry sources report average 2024 market prices around \$4 to \$6 per ton of CO$_2$ equivalent for voluntary offsets \citep{CarbonCredits2025}, with higher rated offsets (as determined by external rating companies) commanding higher prices that are still typically far below the social cost of carbon \citep{Sylvera2026}. This suggests that, all else equal, increases in offset prices may be welfare-improving due to their low initial price.} 
Instead of using conventional carbon accounting as the primary metric, policymakers concerned about social welfare could instead use the marginal value of public funds (MVPF), since the MVPF is a good proxy for welfare change in this setting. 

Second, conventional carbon accounting is also not a good proxy measure for emissions changes, and it can over- or under-credit offsets relative to actual emission reductions. Furthermore, even when all offsets are additional by traditional measures, aggregate emissions can increase due to the general equilibrium effects given an offset price increase. This ``backfire'' outcome can occur via two separate mechanisms. For the first mechanism, the price of energy relative to the price of capital increases, leading to substitution away from (relatively) clean energy and towards emissions-producing capital in the final goods sector.\footnote{Recall that ``capital'' is a composite of all inputs. By way of example, suppose an offset policy leads to a relative increase in the price of energy services for an electric heat pump relative to a natural gas furnace, then consumers shift away from the electric heat pump and towards the emissions-producing natural gas furnace.} 
This first backfire mechanism requires that capital and energy are relatively substitutable in the final goods sector. 
For the second mechanism, if renewable energy and fossil fuels are highly complementary, then an increase in renewable output leads to a corresponding increase in fossil fuel use, potentially creating backfire.\footnote{To illustrate, increases in renewable energy may lead to an increased need for ramping of fossil fuel plants during the times of day when renewables go offline (e.g., sunset); because ramping is relatively more emissions intensive \citep{malik2026}, this can erode some of the emissions benefits of renewables, or even reverse the effect.}

The general equilibrium effects of offsets vary across time and space, as the outcomes described above depend on specific ranges of parameter values---such as the substitutability of clean and dirty energy production---which change over time and by location. For example, when penetration rates of renewable energy are low in electricity markets, renewable and fossil fuel energy can be relatively substitutable, even in the absence of energy storage technologies. Indeed, in the absence of energy storage, intermittency is relatively easy for grid managers to deal with up until 25-50\% penetration of renewables \citep{heal2017reflections, leonard2020energy}. As the penetration of renewables in the energy grid increases, their intermittency becomes more challenging for grid managers, decreasing the substitutability of the two energy sources (lowering $\sigma_E$). However, as energy storage technologies become more cost-effective, renewable intermittency will become less of a concern, and thus renewables will become more substitutable with fossil fuels (raising $\sigma_E)$.
Finally, in the absence of backfire described above, the same mechanisms can still erode some of the emissions benefits and lead to the over-crediting of carbon offsets.

When deciding which sectors, industries, or regions to target with offset payments, policymakers may want to consider settings where $\sigma_E$ is likely to be high (implying a high degree of substitutability between $F$ and $R$), as higher $\sigma_E$ tends to lead to higher welfare in our model. For instance, they may wish to target offset payments to U.S. states with lower penetration rates of renewable energy. Conversely, consider the case of forestry offsets, where $E$ is forestry services, $F$ is timber production, and $R$ is non-consumptive forestry activities (such as recreation). In this setting, $\sigma_E$ is likely to be quite low due to the difficulty of substituting $R$ for $F$, and thus leading to over-crediting or backfire. This implies that forestry offsets may have lower social net benefits than offsets from other industries.

We model carbon offsets as a type of subsidy to renewable energy output. Thus, in addition to carbon offsets, our model can be applied to other green production credits, such as Renewable Energy Credits and the Renewable Energy Production Tax Credit in the U.S., and our results can be extrapolated to those settings as well. In contrast, some existing (non-offset) policies subsidize green capital, $K_R$. For example, U.S. policies also subsidize renewable energy investment via the Residential Clean Energy Tax Credit and the Advanced Energy Project Credit. In our model, the inclusion of the fixed factor of production makes a distinction between a subsidy on output ($R$) and a subsidy on green capital ($K_R)$. However, when $\sigma_R$ is small (i.e., highly complementary), the difference is likely to be small. Thus, if capital and the fixed factor are relatively complementary in green energy production, then many of our results are likely to apply to a green capital subsidy setting. However, our model abstracts away from some of the real-world complexities that cause disparate outcomes between green output and green investment subsidies (see \cite{aldy2023investment} for a more complete discussion of these differences).

We do not model voluntary offset markets in this study. Instead, we model government-purchased offsets in order to cleanly identify the general equilibrium effects and precisely decompose the different sources of emissions changes, offset types, and carbon accounting metrics. Modeling voluntary offsets would distract from the key general equilibrium impacts that we identify in our analysis. Furthermore, the zero-profit conditions mean there are no funds available for fully voluntary offset purchases. Similarly, the effects of offsets tied to specific carbon markets can depend on the exact policy details of the prevailing regulation, and thus we abstract from particular policy design to maintain focus on the overall general equilibrium effects from offsets. Future research could examine voluntary offsets and offsets tied to specific regulatory carbon market regimes in general equilibrium.

Our numerical analysis, including a case study of U.S. carbon capture and storage (CCS) policy, provides intuition for our results and highlights how, even in a relatively simple analytical model, additionality is not equivalent to changes in welfare. A more detailed set of results using a computational general equilibrium model with many disaggregated production and consumption goods could generate additional guidance to policymakers regarding specific offset policies. Nonetheless, our model provides the theoretical foundation for future empirical and numerical studies on the aggregate effects of offset policies in general equilibrium.


\bibliography{offsets.bib}
\bibliographystyle{aer}

\vspace{1cm}
\onehalfspacing
\newpage
\appendix 

\noindent \textbf{{\Large APPENDIX}}


\section{Appendix: Log-Linear Equations} \label{App:LogLinear}

\renewcommand{\theequation}{A-\arabic{equation}}
\setcounter{equation}{0}

This appendix derives the log-linear equations that appear in Section \ref{Sec:Model}.

To begin, totally differentiating the primary-factor resource constraint finds:
\begin{align*}
    d[\bar K &= K_X+K_R+K_F]
    \Rightarrow 0 = dK_X+dK_R+dK_F \\
    \Leftrightarrow 0 &= \frac{K_X}{\bar K}\frac{dK_X}{K_X}+\frac{K_R}{\bar K}\frac{dK_R}{K_R}+\frac{K_F}{\bar K}\frac{dK_F}{K_F}
    \Leftrightarrow 0 = \alpha_X \hat K_X+ \alpha_R\hat K_R+\alpha_F\hat K_F, 
\end{align*}
where $\alpha_i$ is the share of primary factors used in Sector $i$. Note that all parameter definitions are found in Section \ref{Sec:Model}.

We continue with Sector $X$. Totally differentiate the production function and simplify as follows:
\begin{align*}
    d[X&=X(K_X,E)]
    \Rightarrow dX = \frac{\partial X}{\partial K_X}dK_X+\frac{\partial X}{\partial E}dE \\
    \Leftrightarrow dX &= \frac{(1+t)P_K}{P_X}dK_X+\frac{P_E}{P_X}dE \ [\textrm{by profit-max FOCs}]
    \Leftrightarrow \frac{dX}{X} = \frac{(1+t)P_KK_X}{P_XX}\frac{dK_X}{K_X}+\frac{P_EE}{P_XX}\frac{dE}{E} \\
    \Leftrightarrow \hat X &= \frac{(1+t)P_KK_X}{P_XX}\hat K_X+\frac{P_EE}{P_XX}\hat{E} 
    \Leftrightarrow \hat X = \theta_{XK}\hat K_X+\theta_{XE}\hat{E},
\end{align*}
and $\theta_{XK}+\theta_{XE}=1$, where $\theta_{XK} \equiv \frac{(1+t)P_K K_X}{P_X X}$ and $\theta_{XE} \equiv \frac{P_E E}{P_X X}$.
Next, totally differentiating the zero-profit condition for Sector $X$ finds:
\begin{align*}
    d[P_X X &= (1+t) P_K K_X + P_E E] \\
    \Rightarrow dP_X X + P_XdX &= dtP_K K_X + (1+t)dP_K K_X + (1+t)P_K dK_X + dP_E E+ P_E dE \\
    \Leftrightarrow dP_X X + P_XdX &= (1+t)P_K K_X \left( \frac{dt}{1+t} + \frac{dP_K}{P_K} +  \frac{dK_X}{K_X} \right) + dP_E E+ P_E dE \\
    \Leftrightarrow \hat P_X + \hat X &= \frac{(1+t)P_K K_X}{P_XX} (\hat t+ \hat P_K + \hat K_X) + \frac{P_E E}{P_XX} ( \hat P_E+ \hat E) \\    
    \Leftrightarrow \hat P_X + \hat X &= \theta_{XK} (\hat t + \hat P_K + \hat K_X) + \theta_{XE} ( \hat P_E+ \hat E),
\end{align*}
where $\hat t = dt/(1+t)$.
Finally, we define the elasticity of substitution in sector $X$ as:
\begin{align*}
    \sigma_X &\equiv -\frac{d(K_X/E)/(K_X/E)}{d((1+t)P_K/P_E)/((1+t)P_K/P_E)} \geq 0.
\end{align*}
Then, differentiating through the definition finds:
\begin{align*}
    \Rightarrow \sigma_X &= -\frac{((dK_XE-K_XdE)/E^2)\cdot(E/K_X)}{((dtP_KP_E+(1+t)dP_KP_E-(1+t)P_KdP_E)/P^2_E)\cdot(P_E/((1+t)P_K))} \\
    \Leftrightarrow \sigma_X &= -\frac{\frac{dK_X}{K_X}-\frac{dE}{E}}{\frac{dt}{1+t}+\frac{dP_K}{P_K}-\frac{dP_E}{P_E}} 
    \Leftrightarrow \sigma_X = -\frac{\hat K_X-\hat E}{\hat t + \hat P_K - \hat P_E} 
    \Leftrightarrow \hat K_X-\hat E = \sigma_X (\hat P_E - \hat t - \hat P_K).
\end{align*}

Similar derivations to those in sector $X$ apply to sectors $E$ and $F$, so those derivation are omitted here. The only differences are that sector $E$ does not employ ``capital'' directly (and thus its inputs are not subject to the tax) and sector $F$ has only one input (and thus there is no substitution in production). However, since sector $R$ has a fixed factor and its output receives the offset payment, then we provide the full set of derivations below. To start, totally differentiate the production function and simply as follows:
\begin{align*}
     d[R&=R(K_R,\bar{Q})] \Rightarrow dR = \frac{\partial R}{\partial K_R} dK_R \ [\textrm{since}\ d\bar{Q}=0] \\
     d R &= \frac{(1+t)P_K}{P_R+s} dK_R \ [\textrm{by profit-max FOCs}] \Leftrightarrow \frac{dR}{R} =\frac{(1+t)P_K K_R}{(P_R+s)R} \frac{dK_R}{K_R} \Leftrightarrow \hat R = \theta_{RK} \hat K_R,
\end{align*}
and $\theta_{RK}+\theta_{RQ}=1$, where $\theta_{RK} \equiv \frac{(1+t)P_K K_R}{(P_R+s) R}$ and $\theta_{RQ} \equiv \frac{P_{Q} \bar Q}{(P_R+s) R}$. For the zero-profits condition, the derivation goes:
\begin{align*}
    d[(P_R+s)R &= (1+t) P_K K_R + P_{Q}\bar{Q}] \\
    \Rightarrow (dP_R+ds)R+(P_R+s)dR &= dtP_K K_R + (1+t)dP_K K_R + (1+t)P_K dK_R + dP_{Q}\bar{Q} \\
    \Rightarrow (P_R+s)R \left( \frac{dP_R+ds}{P_R+s}+\frac{dR}{R} \right)  &= (1+t)P_K K_R\left( \frac{dt}{1+t}+ \frac{dP_K}{P_K} + \frac{dK_R}{K_R} \right) + P_{Q}\bar{Q} \frac{dP_{Q}}{P_{Q}} \\
    \Rightarrow (P_R+s)R \left( \frac{P_R}{P_R+s}\frac{dP_R}{{R}} + \frac{s}{P_R+s}\frac{ds}{s} +\frac{dR}{R} \right)  &= (1+t)P_K K_R\left( \frac{dt}{1+t}+ \frac{dP_K}{P_K} + \frac{dK_R}{K_R} \right) + P_{Q}\bar{Q} \frac{dP_{Q}}{P_{Q}} \\
    \Rightarrow \frac{P_R}{P_R+s} \hat P_R + \frac{s}{P_R+s} \hat s + \hat R &= \frac{(1+t)P_K K_R}{(P_R+s) R} (\hat t+ \hat P_K + \hat K_R ) + \frac{P_{Q}\bar{Q}}{(P_R+s) R} \hat{P}_{Q} \\
    \Rightarrow (1-\gamma)\hat P_R + \gamma\hat s + \hat R &= \theta_{RK} (\hat t + \hat P_K + \hat K_R) + \theta_{RQ} \hat{P}_{Q},
\end{align*}
where $\hat s = ds/s$ is the change in the per unit offset payment and $\gamma \equiv \frac{s}{P_R+s}$ is the share of the initial offset in the total price of $R$. Despite the fixed factor, firms in Sector $R$ want to substitute according to the equation:
\begin{align*}
    \sigma_R &\equiv -\frac{d(K_R/\bar{Q})/(K_R/\bar{Q})}{d((1+t)P_K/P_Q)/((1+t)P_K/P_Q)} \geq 0.
\end{align*}
Applying the differentiation through the definition, with $\bar{Q}$ fixed, then finds:
\begin{align*}
    \hat K_R &= \sigma_R (\hat P_{Q} - \hat{t} - \hat P_K). 
\end{align*}

Totally differentiating the lump-sum rebate $(B)$ and apply government's budget constraint finds:
\begin{align*}
    d[B&=t(P_K\bar{K})-sR] \Rightarrow
    dB =dtP_K\bar{K}+tdP_K\bar{K}+tP_Kd\bar{K}-dsR-sdR \\
    \Leftrightarrow \frac{dB}{B} &= \frac{tP_K\bar{K}}{B} \left( \frac{1+t}{t} \frac{dt}{1+t}+\frac{dP_K}{P_K} \right) - \frac{sR}{B} \left( \frac{ds}{s}+\frac{dR}{R} \right) \\
    \Leftrightarrow \hat B &= \frac{tP_K\bar{K}}{B} \left( \frac{1+t}{t} \hat t + \hat P_K \right) - \frac{sR}{B} \left( \hat s + \hat R \right) 
    \Leftrightarrow \hat B = \kappa_t \left( \psi \hat t + \hat P_K \right) - \kappa_s \left( \hat s + \hat R \right).
\end{align*}
If the government spends all the tax revenue on the offset payment, then $B=0$ and $\kappa_t=\kappa_s$, and
\begin{align*}
    0 &= \psi \hat t + \hat P_K - \hat s - \hat R \Leftrightarrow
    \hat s + \hat R  = \psi \hat t + \hat P_K, 
\end{align*}
where $\psi = \frac{1+t}{t}$ and this parameter links the \textit{ad valorem} capital tax to the per unit offset payment.

We conclude this appendix by deriving the change in total emission in log-linear form as follows:
\begin{align*}
    d[Z &= \xi_F F + \xi_X K_X] \Rightarrow dZ = \xi_F dF + \xi_X dK_X \Leftrightarrow
    \frac{dZ}{Z} = \frac{\xi_FF}{Z} \frac{dF}F + \frac{\xi_X K_X}{Z} \frac{dK_X}{K_X} \\
    \Leftrightarrow \hat Z &= \frac{\xi_FF}{Z} \hat F + \frac{\xi_X K_X}{Z} \hat K_X 
    \Leftrightarrow \hat Z = \rho_F \hat F + \rho_X \hat K_X,
\end{align*}
where $\rho_F+\rho_X=1$. This last equation assists in the additionality and welfare calculations, but is not part of the log-linear general equilibrium system.


\section{Appendix: Decomposing Sector \textit{R}}
\label{App:Non-Additional}
\renewcommand{\theequation}{E-\arabic{equation}}
\setcounter{equation}{0}

This section outlines the extended model described in Section~\ref{Sec:Model}. Here we start with the premise that the initial equilibrium contains non-additional $R$. That is, some proportion of $R$ is not due to the existing offset payments (in other words, it is non-additional by traditional measures). The total quantity of $R$ can then be decomposed as $R=R^A+R^N$ with $R^A$ denoting the additional component of $R$, while $R^N$ denotes the non-additional component. Totally differentiating this sum leads to:
\begin{align}
    \hat{R} = \phi \hat{R}^A + (1-\phi) \hat{R}^N
    \label{Eq:DecompR}
\end{align}
where $\phi \equiv \frac{R^A}{R}$, or the initial share of $R$ that is additional.

All additional and non-additional offsets in $R$ receive the offset payment and thus the zero-profits condition is $(P_R+s)R=P_{RA}R^A+P_{RN}R^N$. Totally differentiating this condition finds:
\begin{align}
    (1-\gamma)\hat{P}_R + \gamma\hat{s} + \hat{R} = \theta_{RA}(\hat{P}_{RA}+\hat{R}^A) + \theta_{RN}(\hat{P}_{RN}+\hat{R}^N)
\end{align}
where $\theta_{RA}\equiv\frac{P_{RA}R^A}{P_RR}$ and $\theta_{RN}\equiv\frac{P_{RN}R^N}{P_RR}$ are the shares of additional and non-additional offset input costs in sector $R$, respectively, and $\theta_{RA}+\theta_{RN}=1$. 

The production function for additional offsets employs capital and contains a fixed factor, and thus is given by $R^A=R^A(K_A,\bar{Q}_A)$. Totally differentiating recovers:
\begin{align}
    \hat{R}^A = \theta_{AK}\hat{K}_A
\end{align}
since $\bar{Q}_A$ is fixed and $\theta_{AK}\equiv\frac{(1+t)P_K K_A}{P_{RA}R^A}$ is the input cost share for capital in the additional sub-sector.\footnote{We assume that $\bar{Q}_A$ and $\bar{Q}_N$ are fixed; this is justified by the fact that the existing physical capital (e.g., wind turbines or solar panels) is installed on these pieces of land and cannot easily be moved across sectors.} Also, this implies decreasing marginal returns to $K_A$. The sub-sector $R^A$ has a zero-profits conditions too given by $P_{RA}R^A=(1+t)P_K K_A + P_{QA}\bar{Q}_A$, and totally differentiating it returns:
\begin{align}
    \hat{P}_{RA} + \hat{R}^A = \theta_{AK} (\hat t + \hat P_K + \hat K_A) + \theta_{AQ} \hat{P}_{QA}
\end{align}
where $\hat{P}_{QA}$ is the fixed factors price in the additional offset sub-sector and $\theta_{AQ}\equiv\frac{P_{QA}\bar{Q}^A}{P_RR}$ is the share of input costs to the fixed factor. Finally, despite the fixed factor, substitution between the two inputs in sub-sector $R^A$ is represented by $\sigma_R$, and differentiating the definition finds:
\begin{align}
    \hat{K}_A &= \sigma_R \left( \hat{P}_{QA} - \hat{t} - \hat{P}_K \right)
\end{align}

Similarly, define production, and the zero profits, and substitution in sub-sector $R^N$, and differentiate to find:
\begin{align}
    \hat{R}^N &= \theta_{NK}\hat{K}_N \\
    \hat{P}_{RN} + \hat{R}^N &= \theta_{NK} (\hat t + \hat P_K + \hat K_N) + \theta_{NQ} \hat{P}_{QN} \\
    \hat{K}_N &= \sigma_R \left( \hat{P}_{QN} - \hat{t} - \hat{P}_K \right)
\end{align}
where the substitution elasticity $(\sigma_R)$ is the same across sub-sectors.

Capital is employed in four production functions now, so define $K_R=K_A+K_N$, and differentiating yields:
\begin{align}
    \hat{K}_R=\lambda_A \hat{K}_A + \lambda_N \hat{K}_N
\end{align}
where $\lambda_A\equiv\frac{K_A}{K_R}$ and $\lambda_N\equiv\frac{K_N}{K_R}$ with $\lambda_A+\lambda_N=1$, and it still holds that $0 =\alpha_X \hat K_X + \alpha_R \hat K_R + \alpha_F \hat K_F$.

We assume the additional and non-additional sub-sectors have the same cost functions but allow each sub-sector to have a different endowment of the fixed resource. Regardless, perfect substitution $(R=R^A+R^N)$ and positive initial quantities for $R^A$ and $R^N$ means $P_{RA}=P_{RN}=P_R$. These assumptions also imply the input cost shares are identical such that $\theta_{AK} = \theta_{NK} = \theta_{RK}$ and $\theta_{AQ} = \theta_{NQ} = \theta_{RQ}$. In addition, $\phi=\theta_{RA}$ and $(1-\phi)=\theta_{RN}$, and further algebra reveals $\phi=\lambda_A$ and $(1-\phi)=\lambda_N$, too.

As a result, when the additional and non-additional sub-sectors are identical in all economic respects, it must be that the sub-sectors scale perfectly with $R$; that is, $\hat{R}=\hat{R}^A=\hat{R}^N$. In other words, both additional and non-additional offsets change in equal proportion when the offset price changes. The Supplemental Appendix analytically confirms this result. Note that the condition $\hat{R}=\hat{R}^A=\hat{R}^N$ satisfies the general case of $\hat{R} = \phi \hat{R}^A + (1-\phi) \hat{R}^N$.) 


\section{Appendix: Welfare Deviation} \label{App:Welfare}

\renewcommand{\theequation}{B-\arabic{equation}}
\setcounter{equation}{0}

This appendix derives the welfare formula in Section \ref{Sec:Welfare}. To begin, totally differentiate the utility function:
\begin{align*}
    dU = \frac{\partial U}{\partial X} dX + n \frac{\partial U}{\partial Z}dZ.
\end{align*}
Next, substitute the first-order conditions (FOCs) from the representative consumer's utility maximization problem, $\frac{\partial U}{\partial X}=\lambda P_X$, where $\lambda$ is the marginal utility of income (i.e., the Lagrange multiplier on the budget constraint) to get:
\begin{align*}
    dU = \lambda P_X dX + n \frac{\partial U}{\partial X}dZ.    
\end{align*}
Continuing, divide through by $\lambda$ and multiply the terms on the right-hand side by appropriate values of one to yield:
\begin{align*}
    \frac{dU}{\lambda} = P_X X \hat X - \mu (nZ) \hat Z ,
\end{align*}
where $\mu \equiv \frac{-1}{\lambda}\frac{\partial U}{\partial Z}>0$ is the marginal environmental damage (MED) from a change in $Z$. Finally, divide through total income $I$ to find equation \ref{Eq:Welfare} noting $P_XX=I$ at the optimum.


\newpage
\section{Appendix: Additional Figures}\label{App:Figures}
\begin{figure}[h!]
    \centering
    \includegraphics[height=0.25\textheight, width=\textwidth, keepaspectratio]{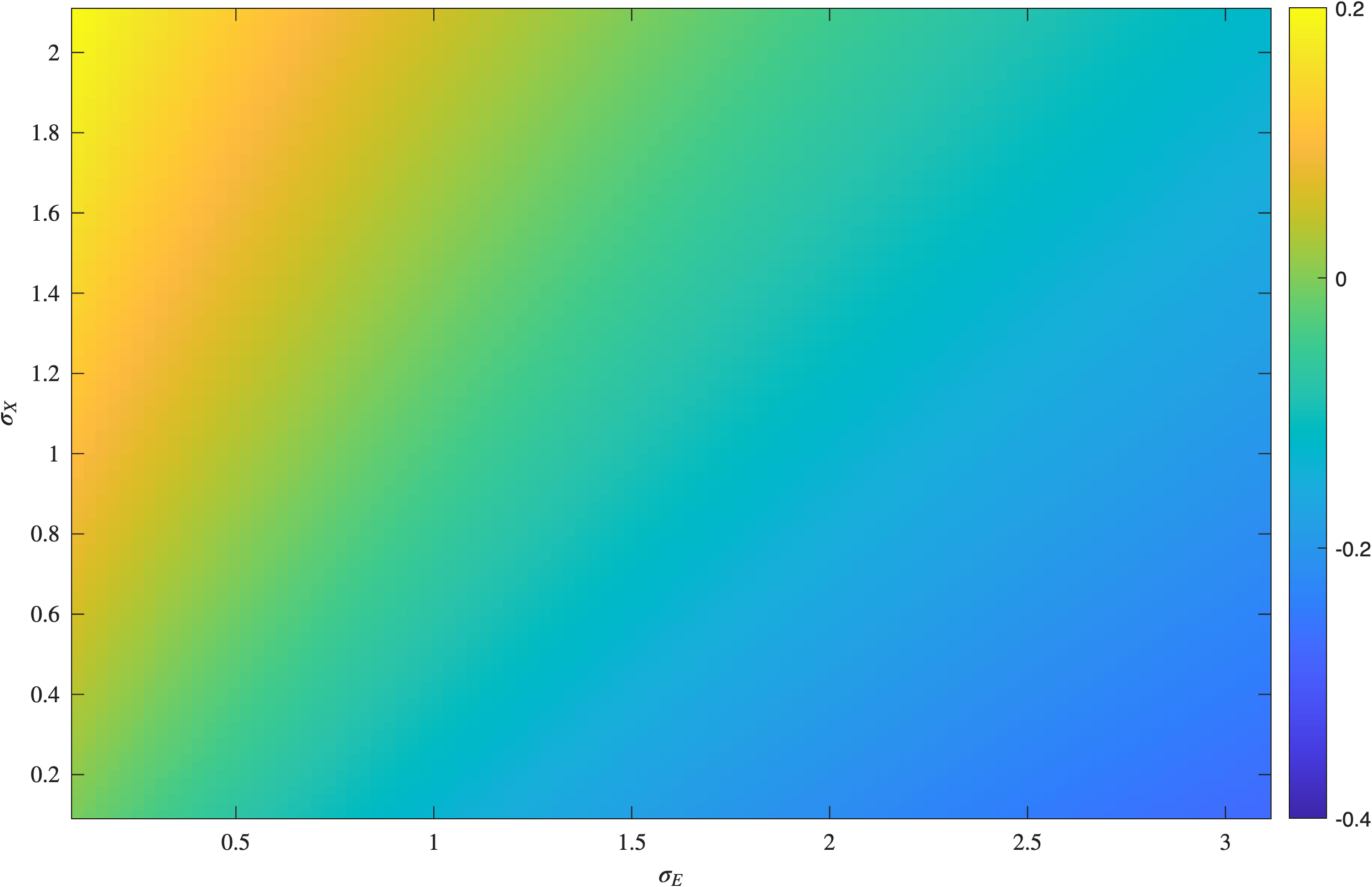}
    \caption{Change in emissions, $\hat{Z}$, as $\sigma_E$ and $\sigma_X$ vary.}
    \vspace{3mm}
    \includegraphics[height=0.25\textheight, width=\textwidth, keepaspectratio]{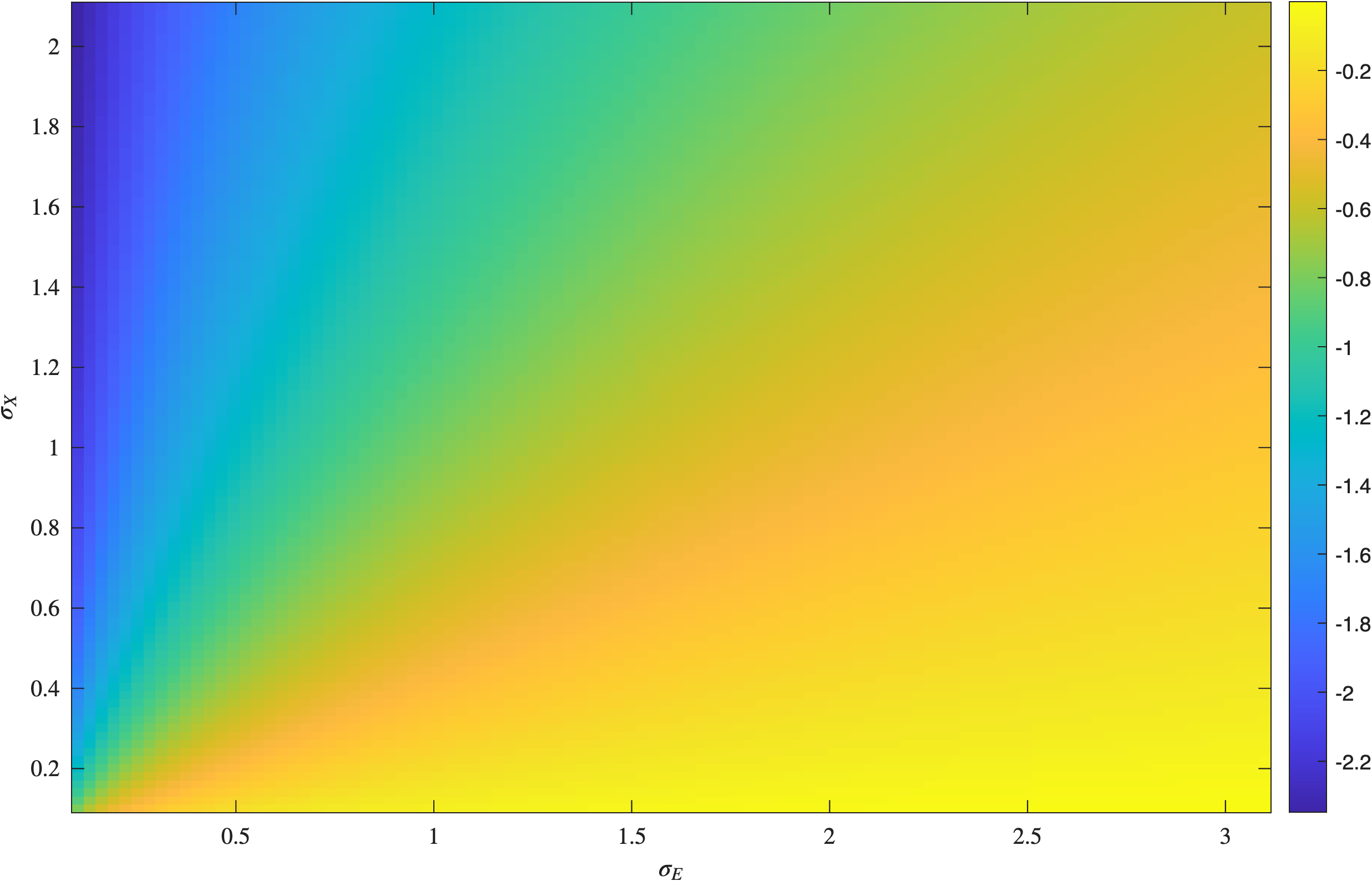}
    \caption{Change in $\Delta$ as $\sigma_E$ and $\sigma_X$ vary.}
    \vspace{3mm}
    \includegraphics[height=0.25\textheight, width=\textwidth, keepaspectratio]{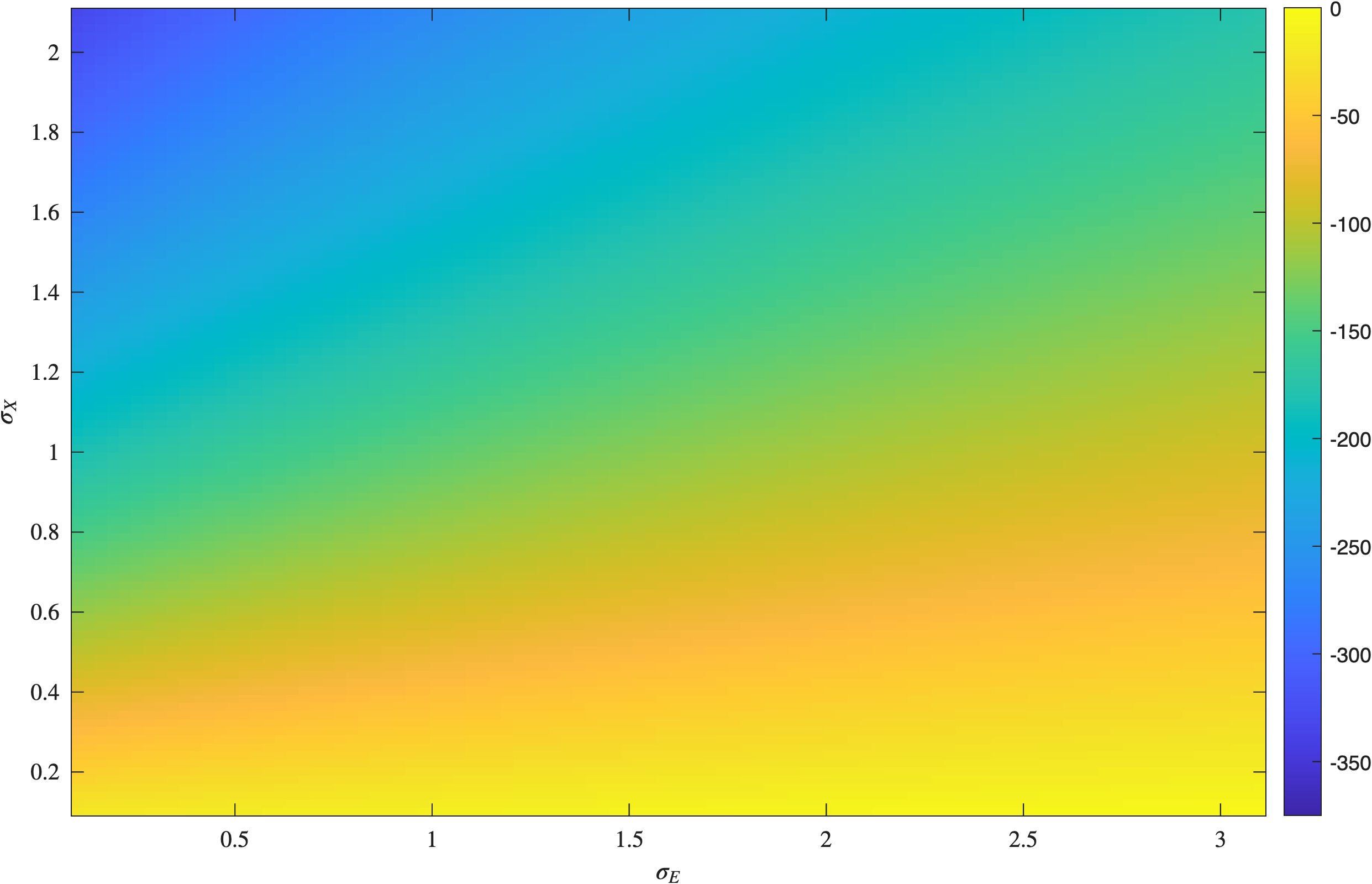}
    \caption{Change in welfare (billions of dollars) as $\sigma_E$ and $\sigma_X$ vary.}
\end{figure}

\end{document}